\documentclass{ws-ijcga}
\usepackage[super]{cite}
\usepackage{url}

\usepackage{graphicx}
\usepackage{subcaption,pdfpages}
\usepackage{amsmath,setspace,amssymb,tikz,graphicx,mathrsfs}

\newcommand{\mF}{\mathcal{F}}
\newcommand{\mG}{\mathcal{G}}
\newcommand{\mH}{\mathcal{H}}
\newcommand{\mO}{\mathcal{O}}
\newcommand{\mP}{\mathcal{P}}
\newcommand{\mR}{\mathcal{R}}
\newcommand{\mS}{\mathcal{S}}
\newcommand{\mT}{\mathcal{T}}
\newcommand{\mU}{AB}
\newcommand{\mV}{AC}
\newcommand{\mW}{BC}
\newcommand{\mX}{\mathcal{X}}
\newcommand{\mY}{\mathcal{Y}}
\newcommand{\mZ}{\mathcal{Z}}

\newcommand{\sW}{\mathscr{W}}

\newtheorem{myproperty}{Property}{\bf}{\it}

\newcommand{\qed}{}

\begin{document}
\markboth{S. Mehrpour, A. Zarei}
{Pseudo-Triangle Visibility Graph: Characterization and Reconstruction}

\catchline

\title{Pseudo-Triangle Visibility Graph: Characterization and Reconstruction}

\author{Sahar Mehrpour\ \ \ \ \ \ \ \  Alireza Zarei}

\address{Department of Mathematical Sciences, Sharif University of Technology}

\maketitle


\begin{abstract}
The visibility graph of a simple polygon represents visibility relations between its vertices.
Knowing the correct order of the vertices around the boundary of a polygon and its visibility graph, it is an open problem to locate the vertices in a plane in such a way that it will be consistent with this visibility graph. This problem has been solved for special cases when we know that the target polygon is a {\it tower} or a {\it spiral}.
Knowing that a given visibility graph belongs to a simple polygon with at most three concave chains on its boundary, a {\it pseuodo-triangle}, we propose a linear time algorithm for reconstructing one of its corresponding polygons.
Moreover, we introduce a set of necessary and sufficient properties for characterizing visibility graphs of pseudo-triangles and propose polynomial algorithms for checking these properties. 
\end{abstract}

\keywords{
Computational geometry; Visibility graph; Characterizing visibility graph; Polygon reconstruction; Pseudo-triangle.
}

\section{Introduction}\label{sec:Introduction}
The {\it visibility graph} of a simple polygon $\mP$ is a graph $\mG(V,E)$ where $V$ is the vertices of $\mP$ and an edge $(u,v)$ exists in $E$ if and only if the line segment $uv$ lies completely inside $\mP$, i.e they are visible from each other.
Based on this definition, each pair of adjacent vertices on the boundary of the polygon are assumed to be visible from each other. This implies that we have always a Hamiltonian cycle in a visibility graph which determines the order of vertices on the boundary of the corresponding polygon.

Computing the visibility graph of a given simple polygon has many applications in computer graphics, computational geometry and robotics. There are several efficient polynomial time algorithms for this problem. Asano~{\it et al.}~\cite{MS10} and Welzl~\cite{MS64} proposed $\mO(n^2)$ time algorithms for computing the visibility graph of a simple polygon of $n$ vertices. This was then improved to $\mO(m+n\log\log n)$ by Hershberger\cite{MS41} where $m$ is the number of edges in the visibility graph. The term $n\log\log n$ is due to the time required for triangulating a simple polygon. Using the $\mO(n)$ time triangulation algorithm of Chazelle~\cite{MS20} reduces the time complexity of Hershberger's result to $\mO(m+n)$ which is optimal. 

This concept has been studied in reverse as well: Characterizing a visibility graph problem is to determine whether a give graph is isomorphic to the visibility graph of some simple polygon, and the reconstruction problem is to build such a simple polygon.
Everett showed that these probleme are in PSPACE\cite{Eve90}, and this is the only result known about the complexity of these problems.
These problems have been solved only for special cases of {\it spiral} and {\it tower} polygons. These results are obtained by Everett and Corneil~\cite{EveCor90} for spiral polygons and by Colley~{\it et al.}~\cite{CoLuSp97} for tower polygons. In spiral polygons there is at most one concave chain (Fig.~\ref{fig:kinds}a) and the boundary of a tower polygon is composed of two concave chains and a single edge (Fig.~\ref{fig:kinds}b).

Although there is a bit progress on this type of reconstruction problem, there have been plenty of studies on characterizing visibility graphs~\cite{AbEK95,CoLuSp97,CouLub92,ElG85,EveCor90,Gh88}. 
In 1988, Ghosh introduced three necessary conditions for visibility graphs and conjectured their sufficiency\cite{Gh88}.
In 1990, Everett proposed a counter-example graph disproving Ghosh's conjecture\cite{Eve90}. She also refined Ghosh's third necessary condition to a new stronger condition\cite{Gh97}.
In 1992, Abello~{\it et al.} built a graph satisfying Ghosh's conditions and the stronger version of the third condition which was not the visibility graph of any simple polygon, disproving the sufficiency of these conditions\cite{AbLinPi92}.
In 1997, Ghosh added his fourth necessary condition and conjectured that this condition along with his first two conditions and the stronger version of the third condition are sufficient for a graph to be a visibility graph. This was also disproved by  a counter-example from Streinu in 2005\cite{Str05}.

There are also several results about other versions of the visibility graph like point visibility graph\cite{SCG00,SCG10,SCG14}, pseudo-polygon visibility\cite{AX,AX12,AX13,AX10} and terrain visibility\cite{MS00}. Excellent surveys on visibility graph results can be found in \cite{MS37,ghosh-book,MS49}.

\begin{figure}
\centering
\includegraphics[scale =0.8]{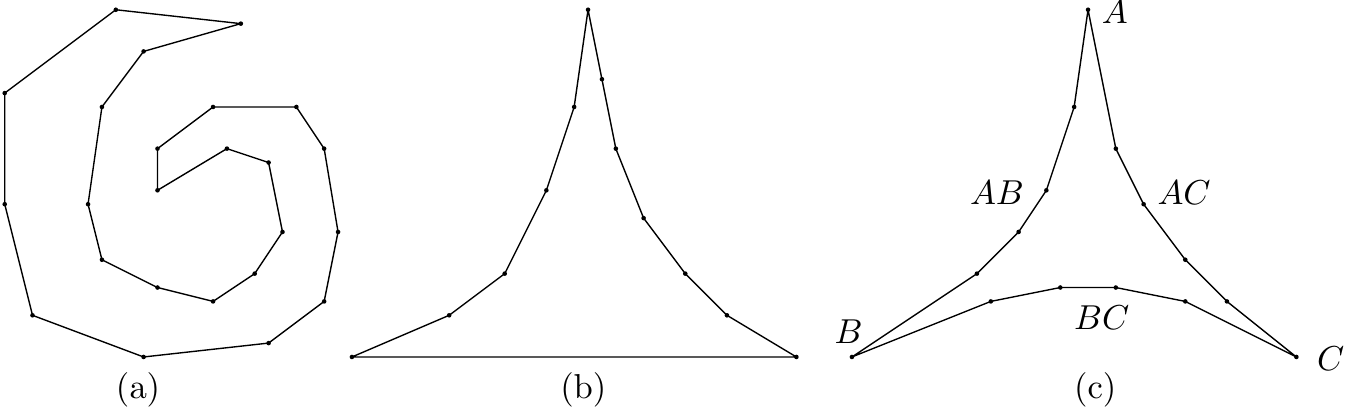} 
\caption{(a) A Spiral polygon, (b) a tower polygon, and (c) a pseudo-triangle.}
\label{fig:kinds}
\end{figure}

In this paper, we solve the reconstruction problem for pseudo-triangles.
A {\it pseudo-triangle} is a simple polygon whose boundary is composed of three concave chains, called {\it side-chains}, where each pair shares one convex vertex (called a corner).
Let $\mP$ be a pseudo-triangle formed by the concave side-chains $\mU=[A,\ldots,B]$, $\mV=[A,\ldots,C]$, and $\mW=[B,\ldots,C]$ where $A$, $B$, and $C$ are the corners (Fig.~\ref{fig:kinds}c). 
According to this notation, a concave side-chain joining corner vertices $X$ and $Y$ is denoted by $XY$ where $X$ and $Y$ are in $\{A,B,C\}$.

Let $\mH=\ <A,\ldots,C,\ldots,B,\ldots,A>$ be the Hamiltonian cycle of the visibility graph of $\mP$ which indicates the order of vertices on the boundary of $\mP$. Here, we use the same notation for a vertex on the boundary of $\mP$ and its corresponding vertex in the visibility graph and $\mH$. For a given pair of Hamiltonian cycle $\mH$ and visibility graph $\mG(V,E)$, we introduce a set of necessary properties on $\mH$ and $\mG$ when this pair belongs to a pseudo-triangle
and prove that these properties are sufficient as well.

Having these properties, we propose a linear-time algorithm for reconstructing a pseudo-triangle 
$\mP=<A,\ldots,C,\ldots,B,\ldots,A>$ with $\mG(V,E)$ as its visibility graph. 
Moreover, we propose algorithms for verifying the properties on a given pair of $\mH$ and $\mG$. These characterizing algorithms run in linear time in terms of the size of $\mG$. Therefore, in this paper we solve the characterizing and reconstructing problems for another class of polygons called pseudo-triangles.

Since a tower polygon is a special case of a pseudo-triangle, we use the tower reconstruction algorithm~\cite{CoLuSp97} as a sub-routine in our algorithm to build the initial part of the polygon. 

Our motivation in solving this problem for pseudo-triangles is that every polygon can be partitioned into pseudo-triangles. Then, an idea for solving a general reconstruction problem is to handle these steps:
\begin{itemize}
\itemsep0.5pt
\item Recognize a pseudo-triangle decomposition for the target polygon from $\mG(V,E)$ and $\mH$.
\item Reconstruct each pseudo-triangle separately.
\item Attach the reconstructed pseudo-triangles satisfying the pseudo-triangle decomposition and the visibility constraints.
\end{itemize}

In Section~\ref{sec:Tower}, we briefly describe the tower reconstruction algorithm~\cite{CoLuSp97} for reconstructing tower polygons which is used as a sub-routine in our algorithm. In Section~\ref{sec:Properties}, we introduce a set of necessary conditions(properties) of the visibility graph of pseudo-triangles and in Section~\ref{sec:reconstruction}, we prove sufficiency of these conditions by proposing a reconstruction algorithm. Finally, we analyze the running time of  the reconstruction algorithm and the algorithms required to check the properties.

\section{Reconstructing Tower Polygons}\label{sec:Tower}
A {\it strong ordering} on a bipartite graph $\mG(V,E)$ with partitions $U$ and $W$ is a pair, $<_U$ and $<_W$, of orderings on respectively ({\it resp.}) $U$ and $W$ such that if $u<_U u'$, $w<_W w'$, and there are edges $(u,w')$ and $(u',w)$ in $E$, the edges $(u,w)$ and $(u',w')$ also exist in $E$.

The following theorem by Colley~{\it et al.}~\cite{CoLuSp97} indicates the main property of the visibility graph of a tower polygon and guarantees the existence of a tower polygon consistent with such a visibility graph.

\begin{theorem}{\cite{CoLuSp97}}\label{thm:tower}
Removing the edges of the reflex chains from the visibility graph of a tower gives an isolated vertex plus a connected bipartite graph for which the ordering of the vertices in the partitions provides a strong ordering. Conversely, any connected bipartite graph with strong ordering belongs to a tower polygon. Furthermore, such a tower can be constructed in linear time in terms of the number of vertices.
\end{theorem}

\begin{figure}
\centering
\includegraphics[scale =1.0]{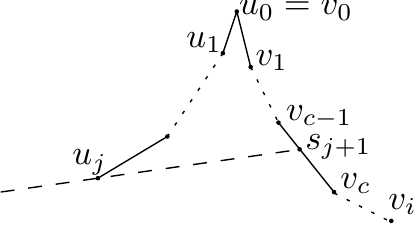} 
\caption{Constructing a tower polygon.}
\label{fig:tower_polygon}
\end{figure}

The outline of the reconstruction algorithm proposed by Colley~{\it et al.}~\cite{CoLuSp97} is as follows. As input, it takes the corner vertex $A=u_0=v_0$ and  a connected bipartite graph $\mG(V,E)$ with vertices partitioned into two independent sets $\mU=\{u_1,\ldots,u_m\}$ and $\mV=\{v_1,\ldots,v_n\}$ having strong ordering. 

In the first step, the position of the corner $A$ and the vertices $u_1$ and $v_1$ are determined as in Fig.~\ref{fig:tower_polygon}. 
In a middle step, suppose that the positions of the vertices $u_0,\ldots,u_{j-1}$ and $v_0,\ldots,v_{k-1}$ have been determined and the directions of the half-lines from $u_{j-1}$ and $v_{k-1}$ which respectively contain $u_j$ and $v_k$, where $(u_j,v_k)\in E$,  are also known. To complete such a middle step, position of $u_j$ and the half-line from $u_j$ which contains $u_{j+1}$ (where $u_{j+1}$ is visible from $v_k$) must be determined.
For this purpose, $u_j$ is located somewhere on its containing half-line horizontally below the vertex $v_c$ which $v_c$ has the minimum index among vertices of $\mV$ which are visible from $u_{j+1}$.
Then, the containing half-line of $u_{j+1}$ will be the half-line on the supporting line of $u_j$ and $s_{j+1}$ downward from $u_j$. Here, $s_{j+1}$ is a point on $v_{c-1}v_c$ with an $\epsilon$ distance below $s_j$, when $s_j$ lies on $v_{c-1}v_c$. If $s_j$ does not lie on $v_{c-1}v_c$, then $s_{j+1}$ is a point on $v_{c-1}v_c$ with an $\epsilon$ distance below $v_{c-1}$. 

According to this construction,
$s_j$ will be the intersection of chain $\mV$ and the supporting line of $u_j$ and $u_{j-1}$. Similarly, $r_j$ will be the intersection of $\mU$ and the supporting line of $v_j$ and $v_{j-1}$ (Fig.~\ref{fig:tower_polygon}). We say {\it ``will''} because we first fix position of $s_j$ (resp. $r_j$) from which position of vertex $u_j$ (resp. $v_j$ ) is determined. We will use this notation once again in Section~\ref{sec:reconstruction}.

\section{Properties of Pseudo-Triangle Visibility Graphs}\label{sec:Properties}

In this section, we describe a set of properties that a pair of $\mH$ and $\mG$ must have to be the Hamiltonian cycle and visibility graph of a pseudo-triangle. 
 
Any sub-sequence $<v_i,\ldots,v_j>$ on the Hamiltonian cycle is called a chain and is denoted by $[v_i,\ldots,v_j]$. 
A vertex $v_a$ on a chain $[v_i,\ldots,v_j]$ is a {\it blocking vertex} for the invisible pair $(v_i,v_j)$ if there is no visible pair of vertices $v_l$ in $[v_i,\ldots,v_{a-1}]$ and $v_k$ in $[v_{a+1},\ldots,v_j]$.
Ghosh showed that for every invisible pair of vertices $(u,v)$ in a visibility graph, there is at least one blocking vertex in $[u,\ldots,v]$ or $[v,\ldots,u]$. Furthermore, every vertex on the shortest Euclidean path between $u$ and $v$(inside the corresponding polygon) is a blocking vertex for this pair~\cite{Gh88}. Note that in a pseudo-triangle the shortest Euclidean path between two invisible vertices turns in only one direction (i.e. clockwise or counterclockwise).

Let $\mU$, $\mV$, and $\mW$ be the side-chains of a pseudo-triangle.  
The order of vertices in these chains is defined with respect to one of their corner vertices. For a vertex $u$ in chain $XY$,
$Ind^{X}(u)$ is equal to the number of vertices in chain $[X,\ldots,u]$ minus one. According to this definition,  $Ind^{A}(A)$ is zero and $Ind^{B}(A)$ is $k-1$ where $k$ is the number of vertices in chain $AB$. Then, based on a given vertex indexing we refer to the previous and next vertices of a given vertex on a side-chain. For a vertex $v$ in chain $XY$ with $Ind^{X}(v)=i$, we use $N^X(v,j)$ to refer to the vertex $u\in XY$ with $Ind^X(u)=i+j$. Similarly, $P^X(v,j)$ is the vertex $u\in XY$ with $Ind^X(u)=i-j$. For the sake of brevity, we use $N^X(v)$ instead of $N^X(v,1)$, and $P^X(v)$ instead of $P^X(v,1)$. Note that in this notation, $j$ can be a positive or negative natural number. For corner vertices which belong to two side-chains we use $N$ or $P$ notation only when the target chain is known from the context. For a vertex $u$, $FV^{X}(u, XY)$ is a vertex on chain $XY$ with minimum index that is visible from $u$ when the indices start from corner vertex $X$. Similarly, $LV^{X}(u, XY)$ is a vertex on chain $XY$ with maximum index that is visible from $u$ when the indices start from corner vertex $X$. We have used $FV$ and $LV$ respectively as abbreviation for {\it first visible} and {\it last visible}. Fig.~\ref{fig:vertices} depicts this notation.

\begin{figure}
\centering
\includegraphics[scale =1.0]{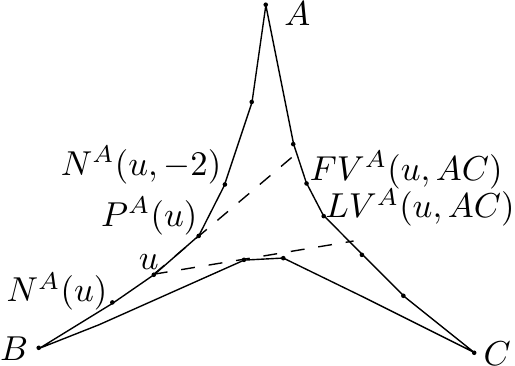} 
\caption{Notation used for vertices.}
\label{fig:vertices}
\end{figure}

\begin{lemma}\label{lem:1}
It is always possible to identify at least two corners of a pseudo-triangle $\mP$ from its corresponding Hamiltonian cycle and visibility graph.
\end{lemma}
\begin{proof}
Since a corner is a convex vertex, it cannot be a blocking vertex for its neighbors. Therefore, in the Hamiltonian cycle of a pseudo-triangle, there are at most three vertices whose adjacent vertices are visible pairs. By traversing the Hamiltonian cycle, these visible pairs, and the corresponding corners, can be identified.

Suppose that this method does not identify all three corners. Without loss of generality, assume that $A$ is an unidentified corner and its adjacent vertices on chains $AB$ and $AC$ are respectively $u$ and $v$. This means that $u$ and $v$ do not see each other and there must be a blocking vertex for this invisible pair. Due to their concavity, this blocking vertex cannot belong to the chains $\mU$ and $\mV$.
Consider the shortest Euclidean path between $u$ and $v$ inside the  pseudo-triangle(Fig.~\ref{fig:pseudo_corners}). It is clear that this path is composed of a subchain of $\mW$, say $[w,\ldots,w']$,  and two edges $(u,w)$ and $(w',v)$ where $Ind^B(w')\ge Ind^B(w)$ and both edges $(u,w)$ and $(w',v)$ belong to the visibility graph.
The polygon formed by $<u,\ldots,B,\ldots,w>$ is a tower polygon with base $(u,w)$ and corner $B$. 
The corner of this tower is the isolated vertex obtained by removing the edges of its Hamiltonian cycle from its visibility graph. Therefore, the corner vertex $B$ is detectable. 
The same argument holds for the tower polygon formed by $<w',\ldots,C,\ldots,v>$ from which the corner $C$ can be identified.
This means that if $A$ cannot be identified from the visibility graph, the other two corners will be detectable. \qed
\end{proof}

\begin{figure}
\centering
\includegraphics[scale =1.0]{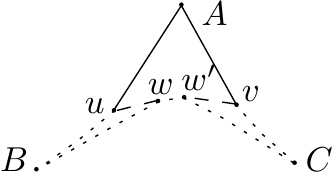} 
\caption{A pseudo-triangle with $B$ and $C$ as its detectable corners.}
\label{fig:pseudo_corners}
\end{figure}

Consider a pseudo-triangle $\mP$ with side-chains $\mU$, $\mV$, and $\mW$, and $\mG$ and $\mH$ as its visibility graph and Hamiltonian cycle, respectively.
Assume that the method described in Lemma~\ref{lem:1}, identifies only two corners of $\mP$. Without loss of generality, assume that $A$ is the unidentified vertex. This means that there is a subchain on $\mW$ which blocks the visibility of adjacent vertices of $A$ on chains $\mU$ and $\mV$. Then, there is no visibility edge between a vertex from $\mU$ and a vertex of $\mV$. By removing the edges of the Hamiltonian cycle from the visibility graph, two isolated vertices $B$ and $C$ and a connected bipartite graph, with parts $S$ and $T$, is obtained where $S$ consists of vertices of chains  $\mU$ and $\mV$ except the isolated vertices $B$ and $C$ and $T$ consists of vertices of $\mW$ except $B$ and $C$. By adding the isolated vertex $B$ to $T$, and the boundary edge $e$ to this bipartite graph that connects $B$ to its adjacent vertex on $\mU$, we will have a single isolated vertex $C$ and a bipartite graph with strong ordering. Then, according to Theorem~\ref{thm:tower} this bipartite graph corresponds to a tower polygon with base edge $e$ and $\mG$ and $\mH$ as its visibility graph and Hamiltonian cycle, respectively.  Fig.~\ref{fig:pseudo2tower} shows how such a pseudo-triangle can be interpreted as a tower polygon.

\begin{figure}
\centering
\includegraphics[scale =1.0]{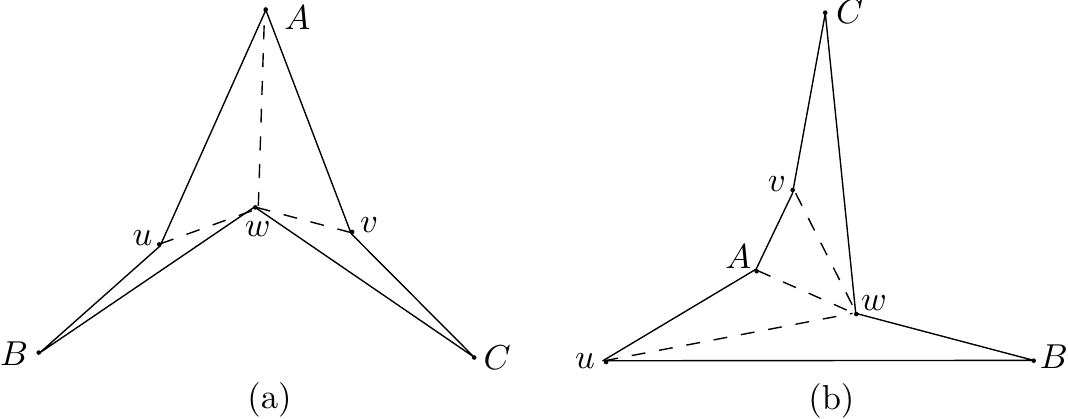} 
\caption{Interpreting a pseudo-triangle as a tower polygon: (a) initial  pseudo-triangle, (b) equivalent tower polygon.}
\label{fig:pseudo2tower}
\end{figure}

Therefore, we have the following property about the pair of $\mH$ and $\mG$ of a pseudo-triangle.

\begin{myproperty}\label{prop:1}
If $\mH$ and $\mG$ are respectively the Hamiltonian cycle and visibility graph of a pseudo-triangle $\mP$, at least two corners of $\mP$ can be identified. Furthermore, if only two corners are detectable, the given $\mH$ and $\mG$ belong to a pseudo-triangle if and only if there is a tower polygon  with $\mH$ and $\mG$ as its Hamiltonian cycle and visibility graph, respectively.
\end{myproperty}

From this property, we assume for the remainder of this section that the method described in the proof of Lemma~\ref{lem:1} identifies all three corners. Because otherwise, we can use the tower polygon algorithm to decide whether the given pair of $\mH$ and $\mG$ belong to a tower polygon( which is a special case pseudo-triangle)  and obtain the answer.

An {\it interval} of a chain with endpoints $p$ and $q$ is the set of points on this chain connecting $p$ to $q$.
Note that in this definition, the endpoints of an interval are not necessarily vertices of the chain. For example, for points $p$ on edge $(u_i,u_{i+1})$ and $q$ on edge $(u_j,u_{j+1})$ of a chain $\mU$ where $i<j$, the interval defined by $p$ and $q$ is the chain $[p,u_{i+1}, \ldots,u_j, q]$.

\begin{myproperty}\label{prop:2}
Every non-corner vertex of a side-chain sees a single nonempty interval from anyone of the other side-chains.
\end{myproperty}
\begin{proof}
The inner angle of such a vertex is more than $\pi$ and its inner visibility region cannot be bounded by a single concave chain. Therefore, it will see some parts from any of the other side-chains.
The continuity of these visible parts on each side-chain is proved by contradiction.
Assume that a vertex $u\in\mU$ sees two disjoint intervals $[v_i,\ldots,v_j]$ and $[v_k,\ldots,v_l]$ from $\mV$ meaning that the interval $(v_j,\ldots,v_k)$ is not visible from $u$. Consider an invisible point $v'$ in $(v_j,\ldots,v_k)$. There must be a blocking vertex for the invisible pair $(u,v')$. 
This blocking vertex must lie on the third side-chain which will also blocks either the visibility of $u$ and $v_j$ or $u$ and $v_k$. \qed
\end{proof}

\begin{myproperty}\label{prop:3}
(Fig.~\ref{fig:prop3Cor1}(a)) For any pair of side-chains $\mU$ and $\mV$ and a pair of vertices $\{u,v\}$ where $u\in\mU$, $v\in\mV$, $v\neq A$, and $u=FV^A(v,\mU)$, we have $(P^A(v),u)\in E$. In other words, the closest vertex to $A$ on $\mU$ which is visible to a vertex $v\in\mV$, for $v\neq A$, is also visible from $P^A(v)$.
\end{myproperty}
\begin{proof}
Consider the subpolygon $<u,P^A(u),\ldots,A,\ldots,P^A(v),v>$. If we triangulate this polygon, there is no internal diagonal connected to $v$ which means that$<u,v,P^A(v)>$ must be a triangle in any triangulation. Therefore, the edge $(u,P^A(v))$ is a diagonal and this edge must exist in the visibility graph. \qed
\end{proof}

\begin{figure}
\centering
\includegraphics[scale =1.0]{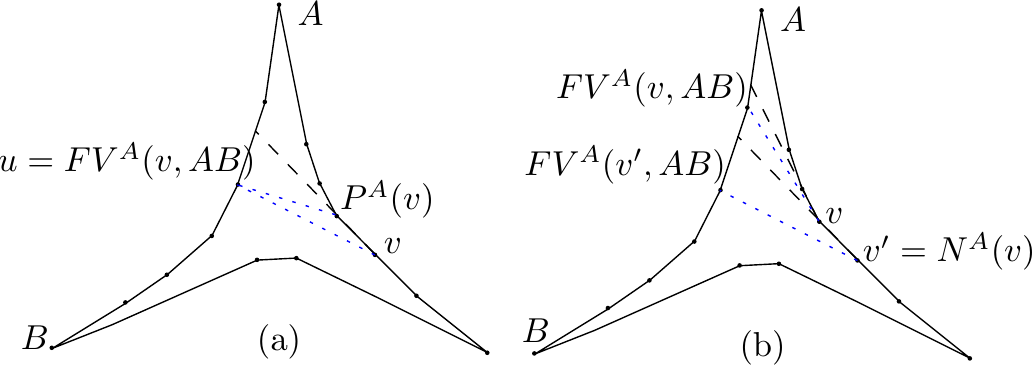} 
\caption{(a) Property~\ref{prop:3}: $P^A(v)$ and $FV^A(v,\mU)$ see each other,  (b) Corollary~\ref{cor:1}: $FV^A(v',\mU)$ cannot be closer to $A$ than $FV^A(v,\mU)$.}
\label{fig:prop3Cor1}
\end{figure}

\begin{corollary}\label{cor:1}
(Fig.~\ref{fig:prop3Cor1}(b)) For any pair of side-chains $\mU$ and $\mV$ and a vertex $v\in\mV$ where $v\neq C$, if $FV^A(v,\mU)=u_j$ and $FV^A(N^A(v),\mU)=u_k$, then \\ $Ind^A(u_k)\ge Ind^A(u_j)$.
\end{corollary}

\begin{corollary}\label{cor:2}
For any pair of side-chains $\mU$ and $\mV$ and a vertex $v\in\mV$ where $v\neq C$, if $v$ does not see any vertex from $\mU$, then $N^A(v)$ does not see any vertex of $\mU$ as well.
\end{corollary}

\begin{corollary}\label{cor:2:2}
(Fig.~\ref{fig:prop4prop5cor3}(a)) For any pair of side-chains $\mU$ and $\mV$ and a pair of vertices $(u,v)$ where $u\in\mU$ and $v\in\mV$ and $k,l>0$, if both $(P^A(u,k),v)$ and $(u,P^A(v,l))$ exist in $E$, then $(P^A(u,k),P^A(v,l))\in E$.
\end{corollary}
\begin{proof}
Lets denote $P^A(u,k)$ and $P^A(v,l)$ by $u'$ and $v'$, respectively. Trivially, $Ind^A(u')\ge Ind^A(FV^A(v,AB))$. Applying Corollary~\ref{cor:1} iteratively on chain $[v',\ldots,v]$ implies that $Ind^A(FV^A(v,AB))\ge Ind^A(FV^A(v',AB))$. This means that $u'$ lies between two vertices $FV^A(v',AB)$ and $u$ which are both visible from $v'$. Then, Property~\ref{prop:2} implies that $u'$ is also visible from $v'$.\qed
\end{proof}

\begin{corollary}\label{cor:2:3}
(Fig.~\ref{fig:prop4prop5cor3}(b)) For any pair of side-chains $\mU$ and $\mV$ and a pair of vertices $(u,v)$ where $u\in\mU$ and $v\in\mV$, if both $(N^A(u,k),N^A(v,l))$ and $(u,v)$ exist in $E$ where $l,k>0$, then at least one of the edges $(N^A(u),v)$ or $(u,N^A(v))$ exists in $E$.
\end{corollary}
\begin{proof}
We prove this by induction on $k+l$. For the induction base step, assume that $k=l=1$. If $v$ is not visible from $N^A(u)$, $N^A(v)$ must be equal to $FV^A(N^A(u),AB)$. Then, Property~\ref{prop:3} implies that $u$ sees $N^A(v)$. 

For the inductive step, assume that the corollary holds for all $k+l<n$ where $n> 2$. Lets denote $N^A(u,k)$ and $N^A(v,l)$ by $u'$ and $v'$, respectively.  If $Ind^A(FV^A(u',AB))\ge Ind^A(v)$, $v$ sees both vertices $u$ and $u'$ which according to Property~\ref{prop:2} sees $N^A(u)$ as well. Otherwise, according to Property~\ref{prop:3}, $FV^A(u',AB)$ is visible from $P^A(u')$. If $P^A(u')=u$, then $u$ sees $v$ and $FV^A(u',AB)$ which is farther from $A$ than $v$ and means that $u$ and $N^A(v)$ see each other. Finally, when $P^A(u')\neq u$ we obtain a smaller version of the problem with parameters $k-1$ and $l$ which holds by induction.\qed
\end{proof}

\begin{figure}
\centering
\includegraphics[scale =.80]{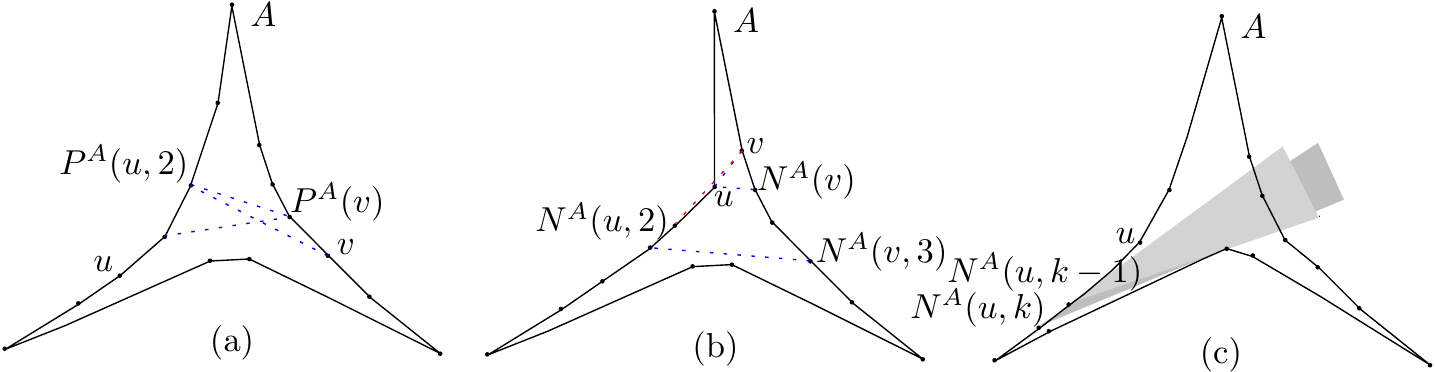} 
\caption{(a) Corollary~\ref{cor:2:2}: $P^A(u,2)$ and $P^A(v)$ must see each other, (b) Corollary~\ref{cor:2:3}: $u$ and $N^A(v)$ see each other, (c) Corollary~\ref{cor:3}: Visible vertices of $\mV$ from $N^A(u,k)$ are also visible to $N^A(u,k-1)$.}
\label{fig:prop4prop5cor3}
\end{figure}

\begin{corollary}\label{cor:3}
(Fig.~\ref{fig:prop4prop5cor3}(c)) For any pair of side-chains $\mU$ and $\mV$ and a pair of vertices $u\in\mU$ and $v\in\mV$, where $(u,v)\in E$ and none of the edges $(N^A(u),v)$ and $(u,N^A(v))$ exist in $E$, all visible vertices of $\mV$ from $N^A(u,k)$ are also visible from $N^A(u,k-1)$ (for any $k>0$). This implies that $LV^A(N^A(u,k),\mV)$ must lie above $v$.
\end{corollary}
\begin{proof}
Any visible vertex $v'$ must belong to $[A,\ldots,P^A(v)]$. Otherwise, according to Corollary~\ref{cor:2:3} either $(N^A(u),v)$ or $(u,N^A(v))$ must exist. According to Corollary~\ref{cor:1}, $FV^A(u,\mV)$ is closer to $A$ than $FV^A(N^A(u,k),\mV)$, and because of the continuity of the chain that is visible from $u$ (Property~\ref{prop:2}) , $v'$ will be visible from $u$. This implies that $v'$ is visible from all vertices of the chain $[u,\ldots,N^A(u,k)]$. \qed
\end{proof}

For each pair of vertices $u\in\mU$ and $v\in\mV$, the diagonal edge $(u,v)$ in the visibility graph of a pseudo-triangle specifies a tower formed by the boundary vertices $<u,\ldots,A,\ldots,v>$. The vertices of this tower satisfy the strong ordering defined earlier. This strong ordering can be derived from Property~\ref{prop:2} and corollaries~\ref{cor:2:2} and \ref{cor:2:3}. Therefore, we do not specify this as a new property.

\begin{myproperty}\label{prop:689}
For any pair of side-chains $\mU$ and $\mV$ and a pair of vertices $u\in\mU$ and $v\in\mV$, where $(u,v)\in E$ and none of the edges $(N^A(u),v)$ and $(u,N^A(v))$ exist in $E$, 
\begin{itemize}
\item[a.] (Fig.~\ref{fig:prop689}(a)) there is a nonempty subchain of the third side-chain $\mW$ which is visible from both $u$ and $v$.
\item[b.] (Fig.~\ref{fig:prop689}(b)) let $[w,\ldots,w']$ be the maximum subchain of $\mW$ visible to both $u$ and $v$ where $w'=N^B(w,l)$, $l\geq0$. Then, $w'$ is not closer to $B$ than $LV^B(N^A(u),\mW)$, or formally, $Ind^B(w')\ge Ind^B(LV^B(N^A(u),\mW))$.
\item[c.] (Fig.~\ref{fig:prop689}(b)) $FV^B(N^A(v), \mW)$ is not closer to $B$ than $LV^B(N^A(u),\mW)$, or formally, $Ind^B(FV^B(N^A(v), \mW))\ge Ind^B(LV^B(N^A(u),\mW))$.
\end{itemize}
\end{myproperty}
\begin{proof}
(a) Triangulating $\mP$ using the edge $(u,v)$, the adjacent triangle of this edge in the opposite side of $A$ must have its third vertex on $\mW$. This is due to the invisibility of $(N^A(u),v)$ and $(u,N^A(v))$ pairs. Therefore, this chain contains at least one vertex. 
From Property~\ref{prop:2} we know that the visible part of $\mW$ from any one of vertices $u$ and $v$ is continuous and the intersection of these parts will be continuous as well.

(b) From (a) we know that the subchain $[w,\ldots,w']$ is nonempty.
For the sake of a contradiction, assume that $w''=LV^B(N^A(u),\mW)$ is farther from $B$ than $w'$. 
Then, the segments $(w'',N^A(u))$ and $(w',v)$ intersect each other inside the pseudo-triangle. Let $p$ be this intersection point. The subpolygon formed by the boundary vertices $<u,N^A(u),p,v>$ must be a convex polygon which completely lies inside the pseudo-triangle. Otherwise, $w'$ will prevent $N^A(u)$ and $w''$ from seeing each other. So, the diagonal edge $(N^A(u),v)$ must exist in $E$ which is a contradiction.

(c) Let $v'$ be $N^A(v)$ and $u'$ be $N^A(u)$. 
For the sake of contradiction, assume that $FV^B(v', \mW)$ is closer to $B$ than $LV^B(u',\mW)$. Then, the edges $(v',FV^B(v', \mW))$ and $(u',LV^B(u',\mW))$ intersect within the pseudo-triangle. Let $p$ be this intersection point. The subpolygon formed by the boundary vertices $<u,v,v',p,u'>$ must be a convex polygon which completely lies inside the pseudo-triangle. Otherwise, $FV^B(v', \mW)$ will prevent $u'$ and $LV^B(u',\mW)$ from seeing each other. So, all diagonal edges $(u,v')$, $(u',v)$, and $(u',v')$ must exist in $E$ which is a contradiction. \qed
\end{proof}

\begin{figure}
\centering
\includegraphics[scale =.80]{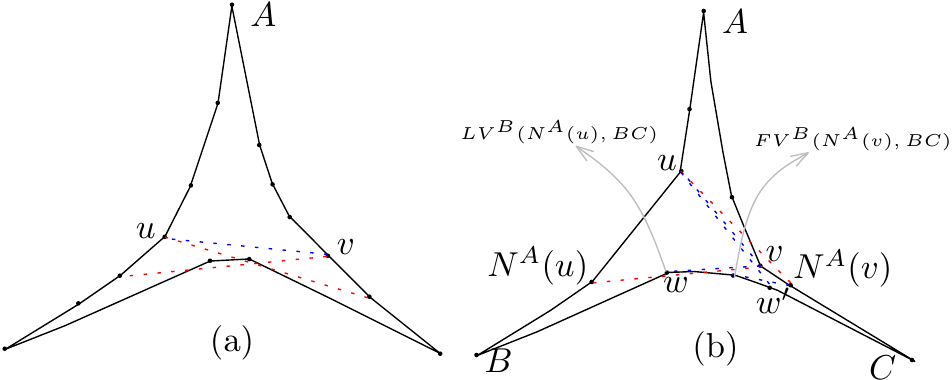} 
\caption{(a) Property~\ref{prop:689}(a): $u$ and $v$ must see common vertices on $mW$, (b) Property~\ref{prop:689}(b): $w'$ is not closer to $B$ than $LV^B(N^A(u),\mW)$, Property~\ref{prop:689}(c): $FV^B(N^A(v),\mW)$ is not closer to $B$ than $LV^B(N^A(u),\mW)$.}
\label{fig:prop689}
\end{figure}

\begin{corollary}\label{cor:4}
For any side-chain $\mW$, there exists at least one vertex $w\in\mW$ that sees some vertices from both of the other side-chains. Furthermore, every vertex $P^B(w,k)$ where $k>0$, sees at least one vertex from $\mU$.
\end{corollary}
\begin{proof}
If there is a pair of vertices $u\in\mU$ and $v\in\mV$ satisfying Property~\ref{prop:689}(a), the first part holds for the vertices of the subchain of the third side-chain $\mW$ which is visible from both $u$ and $v$. If there is no such a pair of vertices, without loss of generality assume that $B$ sees some vertices of $\mV$ and $v=LV^A(B,\mV)$. Trivially, the adjacent vertex of $B$ on side-chain $\mW$ sees both $B\in\mU$ and $v\in\mV$. This can be obtained directly from Property~\ref{prop:689}(a) by imaginary cloning $B$ as two separate vertices on $\mU$ and $\mV$. and adding new corner vertex $B$ as a point on the supporting line of $B$ and $v$ in the opposite side of $v$.

Having a vertex satisfying the first part, the second part follows from Property~\ref{prop:3}. \qed
\end{proof}

\begin{myproperty}\label{prop:7}
(Fig.~\ref{fig:prop7}) For any side-chain $\mW$ and a vertex $w\in\mW$ with distinct vertices $u=FV^A(w,\mU)$ and $v=FV^A(w,\mV)$, the vertices $u$ and $v$ are visible from each other.
\end{myproperty}
\begin{proof}
Let $\mP'$ be the subpolygon with $<A,\ldots,u,w,v,\ldots,A>$ as its boundary vertices. The vertex $w$ does not see any other vertex of $\mP'$ which means that the diagonal $uv$ must be used to triangulate $\mP'$. This means that $u$ and $v$ must be visible from each other. \qed
\end{proof}

\begin{figure}
\centering
\includegraphics[scale =.80]{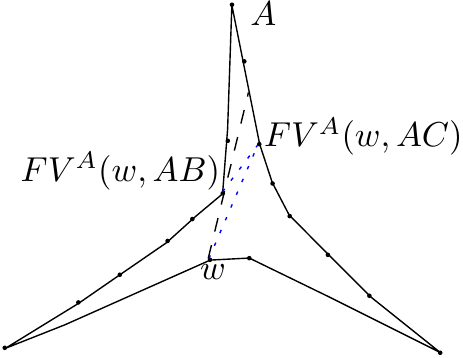} 
\caption{Property~\ref{prop:7}: $FV^A(w,\mU)$ and $FV^A(w,\mV)$ must see each other.}
\label{fig:prop7}
\end{figure}

\begin{myproperty}\label{prop:10}
(Fig.~\ref{fig:prop10}) For any side-chain $\mW$, let $u$ and $v$ be respectively the closest vertices on $\mU$ and $\mV$ to $A$ which are visible from some vertex (not necessarily the same) of $\mW$.
Then, there exists a nonempty subchain $[w,\ldots,w']$ in $\mW$ , $w'=N^B(w,l)$ and $l\geq0$, that either all vertices of this subchain are visible from both $u$ and $v$, or, $(w,w')$ is an edge of $\mW$ and $w$ sees $v$ and $w'$ sees $u$.
\end{myproperty}
\begin{proof}
It is simple to show that $(u,v)\in E$. Assume that there is no vertex on $\mW$ that sees both vertices $u$ and $v$. Then, we first show that there is a pair of vertices $w$ and $w'=N^B(w,l)$ where $w$ sees $v$ and $w'$ sees $u$. Let $w$ be $FV^C(v,\mW)$ and $w'$ be $FV^B(u,\mW)$. Trivially, $w\neq w'$ and $w$ is closer to $B$ than $w'$ (otherwise, $u$ and $v$ will be visible to both $w$ and $w'$). To complete the proof, it is enough to show that $w'=N^B(w)$. This is done by showing that any vertex $w''$ between $w$ and $w'$ on $\mW$ must see at least one of the vertices $u$ and $v$ which contradicts the definition of $w$ and $w'$.

Assume that there is a vertex $w''$ between $w$ and $w'$ and it sees neither $u$ nor $v$. 
In the tower polygon formed by boundary $<u,\ldots,B,\ldots,w'',w'>$, the blocking vertex for the invisible pair $(w'',u)$ must lie on $\mU$. Similarly, in the tower formed by boundary $<w,w'',\ldots,C,\ldots,v>$, the blocking vertex for the invisible pair $(w'',v)$ must lie on $\mV$. Therefore, at least one of the side-chains $\mU$ and $\mV$ must be convex which is a contradiction. So, $w''$ must see at least one of the vertices $u$ and $v$. \qed
\end{proof}

\begin{figure}
\centering
  \begin{tabular}{@{}cccc@{}}
    \includegraphics[width=.3\textwidth]{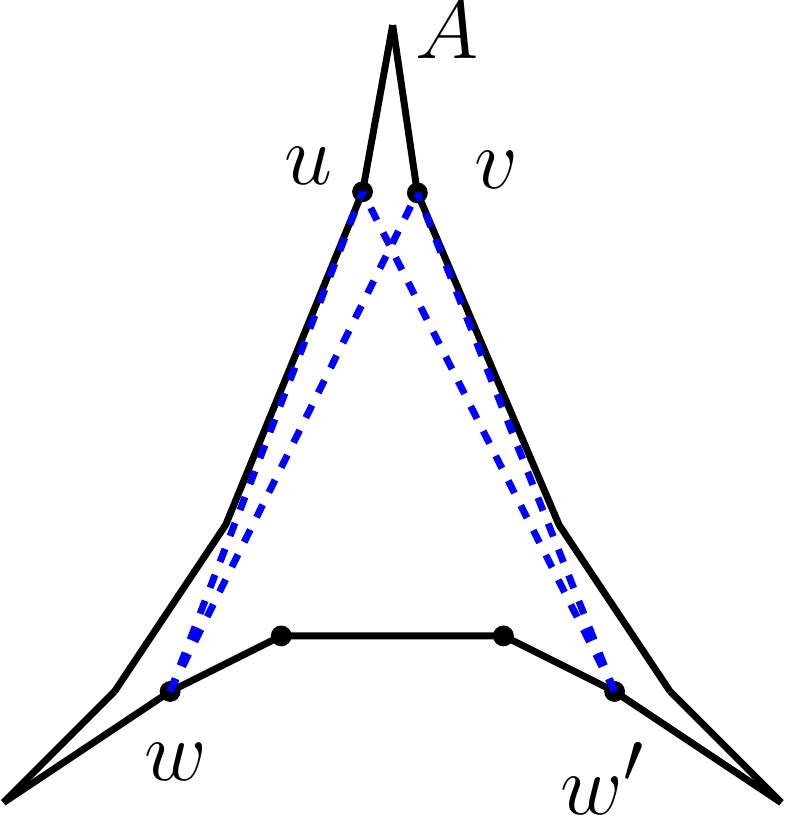} &
    \includegraphics[width=.3\textwidth]{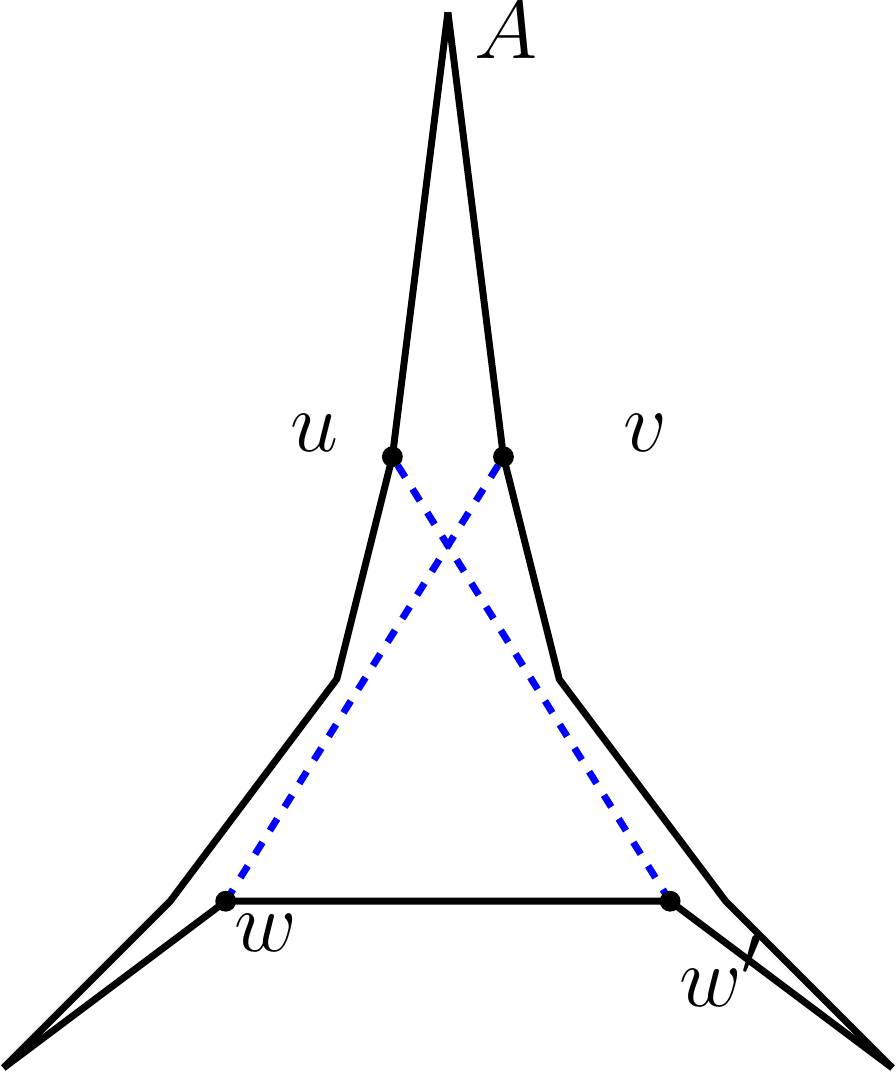} 
  \end{tabular}
\caption{Different cases of Property~\ref{prop:10}.}
\label{fig:prop10}
\end{figure}

\begin{corollary}\label{cor:5}
(Fig.~\ref{fig:cor5}) If $w$ and $w'$ satisfy the conditions of Property~\ref{prop:10}, then for $k>0$:
\begin{itemize}
\item $u_i=FV^A(P^B(w',k),\mU)$ is not closer to $A$ than \\ $u_j=FV^A(P^(w',k-1),\mU)$.
\item If there are vertices $v_i=FV^A(P^B(w,k),\mV)$ and \\ $v_j=FV^A(P^B(w,k-1),\mV)$, then $v_i$ is not closer to $A$ than $v_j$.
\end{itemize}
These mean that as we move from $w'$ to $w$ the topmost visible vertices of $\mU$ and $\mV$ go down along these chains.
\end{corollary}
\begin{proof}
For the sake of contradiction, assume that $u_i$ is closer to $A$ than $u_j$. 
The diagonal edge $(w',FV^A(w',\mU))$ along with vertices $<w',\ldots,B,\ldots,FV^A(w',\mU)>$ form a tower polygon which contains the vertices $u_i$ and $u_j$, and satisfies strong ordering. When both edges $(P^B(w',k),u_i)$ and $(P^B(w',k-1),u_j)$ exist in the visibility graph, the edge $(P^B(w',k-1),u_i)$ must also exist in $E$.

We prove the second part by contradiction. Let $P^B(w, l)$ be the closest vertex of $\mW$ to $B$ which sees at least one vertex from $\mV$ ($l\geq0$). For $l\geq k>0$, assume that 
$v_i$ is closer to $A$ than $v_j$. Since $FV^A(w,\mV)$ is not farther from $A$ than $v_i$, Corollary~\ref{cor:2:2} implies that $v_i$ sees $w$. According to Property~\ref{prop:2}, $v_i$ is also visible from $P^B(w,k-1)$ which is a contradiction. \qed
\end{proof}

\begin{figure}
\centering
\includegraphics[scale =0.45]{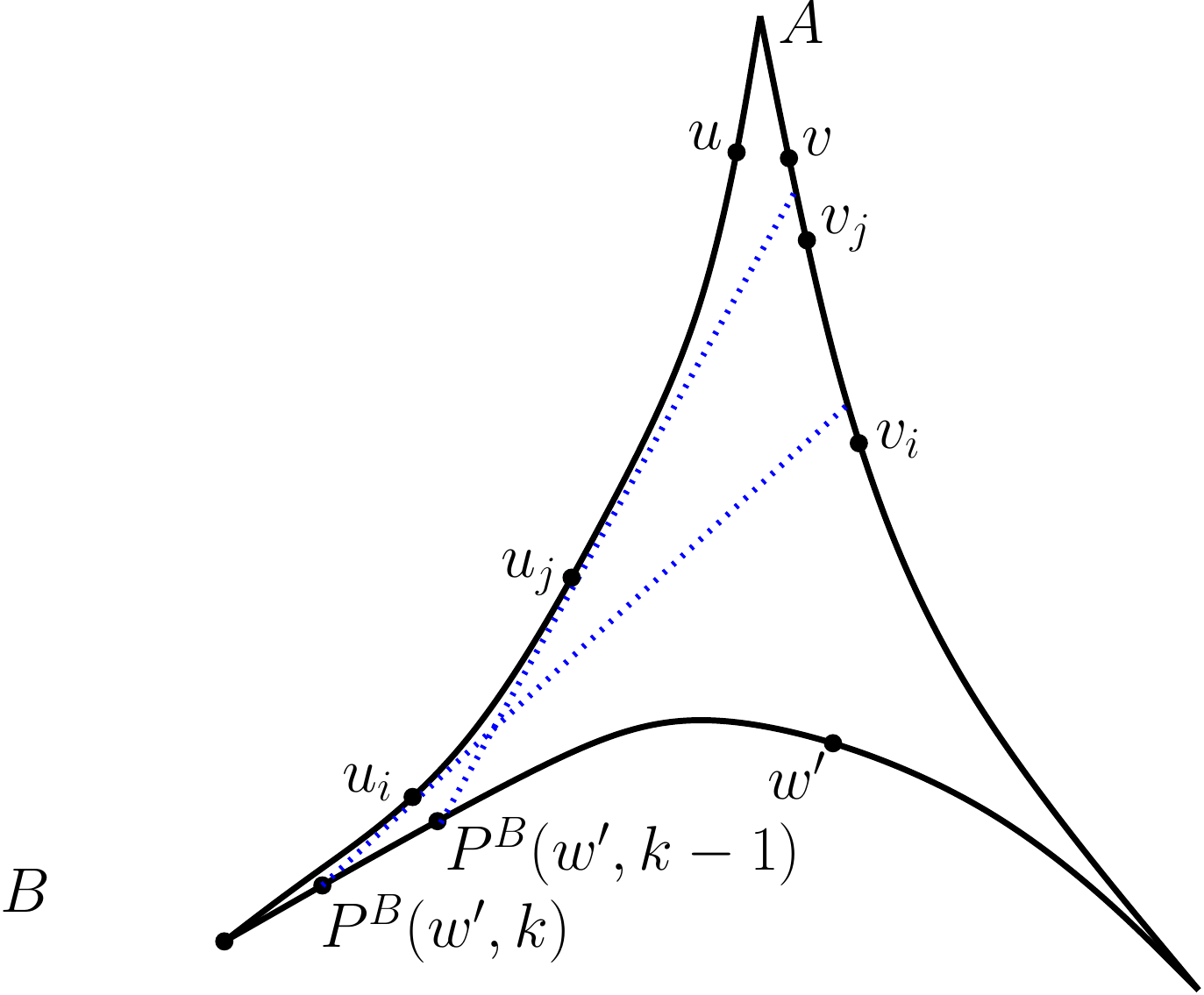} 
\caption{Corollary~\ref{cor:5}: $u_i$ (resp. $v_i$) is not closer to $A$ than $u_j$ (resp. $v_j$).}
\label{fig:cor5}
\end{figure}

\begin{myproperty}\label{prop:11}
Let $[w_i,\dots,w_j]$ be the subchain of $\mW$ satisfying Property~\ref{prop:10} and for any vertex $w\in\mW$, $u=FV^A(w,\mU)$ and $v=FV^A(w,\mV)$ are the closest vertices to $A$ which are visible to $w$. Then:
\begin{itemize}
\itemsep0.5pt
\item[a] If $w\in [w_i,\dots,w_j]$, then at least one of the pairs $(N^A(u),P^A(v))$ and \\ $(P^A(u),N^A(v))$ are invisible.
\item[b] If $w\in [B,\dots,w_j]$ and $(N^A(u),P^A(v))$ are invisible vertices, then this happens for all vertices in $[B,\ldots,w]$ and $LV^A(w,\mV)$ is not farther from $A$ than $LV^A(N^B(w),\mV)$. This is symmetrically true when $w\in [w_i,\dots,C]$ and $(P^A(u),N^A(v))$ are invisible vertices.
\item[c] If $w\neq B$ is closer to $B$ than $w_i$, then $(N^A(u),P^A(v))$ is an invisible pair. Symmetrically, $(P^A(u),N^A(v))$ are invisible vertices when $w\neq C$ is closer to $C$ than $w_j$.
\end{itemize}
\end{myproperty}
\begin{proof}
(a) Consider the subpolygon $\mP'$ with boundary $<w,v,\ldots,A,\ldots,u>$. The pairs $(w,P^A(v))$ and $(w,P^A(u))$ are invisible. These pairs share the same blocking vertex. If $u$ is the blocking vertex, then $(N^A(u),P^A(v))$ is an invisible pair, and if $v$ is the blocking vertex, then the pair $(P^A(u),N^A(v))$ is invisible. 

(b) Assume that $(N^A(u),P^A(v))$ are invisible from each other. This means that the visible vertices of $\mV$ from $w$ is bounded from above by vertices of $\mU$. This will happen for all vertices in $[B,\ldots,w]$ as well. A similar argument holds when $(P^A(u),N^A(v))$ is the invisible pair. 

(c) It is clear that at least one of the vertices $FV^A(w_i,\mU)$ and $FV^A(w_i,\mV)$ is farther from $A$ than $u$ and $v$.
For the sake of a contradiction, assume that $(N^A(u),P^A(v))$ is a visible pair. Then, in subpolygon $\mP'=<w,v,\ldots,A,\ldots,u>$, $v$ must be the blocking vertex for the pairs $(w,P^A(v))$ and $(w,P^A(u))$. This vertex also blocks the pairs $(N^B(w),P^A(u,i))$ and $(N^B(w,l),P^A(v,j))$. But, for some $l>0$ and $i$ and $j\geq0$, $N^B(w,l)=w'$, $P^A(u,i)=FV^A(w',\mU)$, and $P^A(v,j)=FV^A(w_i,\mV)$ which contradicts the definition of $w_i$. \qed
\end{proof}

As mentioned earlier, Ghosh introduced four necessary conditions for a visibility graph of a simple polygon. It is simple to show that these conditions are derived from the properties described in this section which meas that these properties includes Ghosh's conditions.

\section{Pseudo-Triangle Reconstruction}\label{sec:reconstruction}
In this section, $\mG(V,E)$ denotes the visibility graph of a pseudo-triangle $\mP$ with $\mU$, $\mV$, and $\mW$ side-chains and the order of vertices on the boundary of $\mP$ is specified by a Hamiltonian cycle $\mH=<A,\ldots,C,\ldots,B,\ldots,A>$ in $\mG$. 
We assume that the inputs $\mG$ and $\mH$ satisfy the properties~\ref{prop:1} to~\ref{prop:11}. 
We propose an algorithm for reconstructing a pseudo-triangle corresponding to the given pair of $\mG$ and $\mH$.

In order to reconstruct the pseudo-triangle $\mP$, we divide $\mP$ into four subpolygons $\mX$, $\mY$, $\mZ$, and $\mZ'$ as shown in Fig.~\ref{fig:part} and reconstruct each one separately. For the sake of brevity, $u_i=N^A(A,i)$ on side-chain $\mU$, $v_j=N^A(A,j)$ on side-chain $\mV$, and $w_k=N^B(B,k)$ on side-chain $\mW$ where $i,j,k\geq0$. We assume that $\mU$ and $\mV$ have respectively $\alpha+1$ and $\delta+1$ vertices.

\begin{figure}
\centering
\includegraphics[scale =1.0]{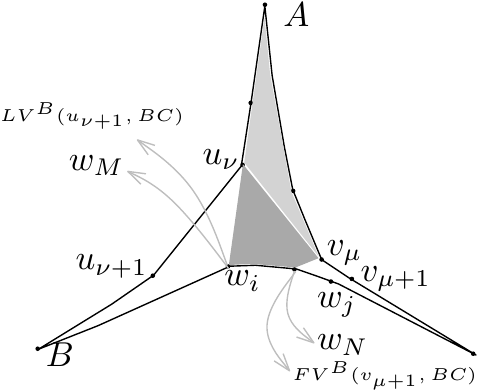} 
\caption{The partitions of the initial polygon in reconstruction algorithm: the light-gray region is $\mX$, the dark-gray is $\mY$ and the white parts are $\mZ$ and $\mZ'$.}
\label{fig:part}
\end{figure}

The subpolygon $\mX$ is formed by subchains $[A,\ldots,u_\nu]$ and $[A,\ldots,v_\mu]$ and edge $(u_\nu,v_\mu)$ where $LV^A(u_\nu,\mV)=v_\mu$ and $LV^A(v_\mu,\mU)=u_\nu$. The vertices $u_\nu$ and $v_\mu$ are identified by walking alternatively on side-chains $\mU$ and $\mV$ from corner vertex $A$ towards $B$ and $C$. As a step of this trace, assume that we are at vertices $u_i$ and $v_j$ and want to go one step further on $\mU$. If $u_i$ is the last vertex on $\mU$ or $v_j$ does not see $u_{i+1}$ we fix $u_i$ as $u_\nu$. Otherwise, we go to $u_{i+1}$ in this step. Walking on side-chain $\mV$ are done similarly. The subpolygon $\mX$ is a tower polygon with strong ordering in its visibility graph. 
Note that $u_{\nu+1}$ or $v_{\mu+1}$ exists only when the side-chain $\mW$ has more than one edge, otherwise, two identified adjacent corners $u_\nu$ and $v_\mu$ compose the base of a tower polygon which can be constructed by the tower reconstruction algorithm. So, we assume that $\mW$ has more than one edge.

The subpolygon $\mY$ is identified as follows: Let $[w_i,\ldots,w_j]$ be the maximum subchain of $\mW$ visible from both $u_\nu$ and $v_\mu$. According to Property~\ref{prop:689}(a), this chain is nonempty and continuous. Let $LV^B(u_{\nu+1},\mW)=w_k$ and $FV^B(v_{\mu+1},\mW)=w_l$. From Property~\ref{prop:689}(b), $k\leq j$ and $l\geq i$ and from Property~\ref{prop:689}(c), $k\leq l$. We define $M$ and $N$ as $\max(k,i)$ and $\min(l,j)$, respectively. It is clear that chain $[w_M,\ldots,w_N]$ contains at least one vertex. Then, $\mY$ is defined to be the polygon with $<u_\nu,w_M,\ldots,w_N,v_\mu>$ as its boundary.

The subpolygon $\mZ$ is formed by subchains $[u_\nu,\ldots,B)]$ and $[B,\ldots,w_M]$ and edge $(u_\nu,w_M)$. Similarly, subchains $[v_\mu,\ldots,C]$ and $[C,\ldots,w_N]$ and edge $(v_\mu,w_N)$ specify the subpolygon $\mZ'$. It is clear that $\mP$ is the union of $\mX$, $\mY$, $\mZ$, and $\mZ'$.

Our reconstruction algorithm first builds $\mX$ using the tower reconstruction algorithm in such a way that vertices of $\mU$ lie to the left of vertices of $\mV$.
Then, we extend this polygon to build $\mY$~(Section~\ref{sec:recontructY}) and build and attach $\mZ$ and $\mZ'$ parts to this polygon~(Section~\ref{sec:recontructZ}) to complete the construction procedure.
\vspace{-2mm}
\subsection{Reconstructing $\mY$}\label{sec:recontructY}
\vspace{-2mm}
In this step, we build the subpolygon $\mY=<u_\nu,w_M,\ldots,w_N,v_\mu>$. We know the position of vertices $u_\nu$ and $v_\mu$ from the previous step, which are also on the boundary of $\mY$. To locate positions of other vertices, we show that there are nonempty regions in which these vertices can be placed.

For any vertex $w_j$ from $\mY$ which $FV^A(w_j,\mU)=u_i$ and $FV^A(w_j,\mV)=v_l$, we define a region $\sW^j_{i,l}$ from which each point sees all vertices in the subchains $[u_i,\ldots,u_\nu]$ and $[v_l,\ldots,v_\mu]$. Therefore, $w_j$ can be placed in $\sW^j_{i,l}$ satisfying the visibility constraints between $w_i$ and vertices of $\mX$. We use $\sW^j$ instead of $\sW^j_{i,l}$ whenever $i$ and $l$ indices are not important. The region $\sW^j$ is determined as follows: Since $w_j$ sees $u_\nu$ and $v_\mu$, the vertices $u_i$ and $v_l$ always exist and are well-defined. 
If $u_i$ and $v_l$ are identical, then $i=l=0$ and the region $\sW^j=\sW^j_{0,0}$ is defined to be the part of the cone formed by the lines through $(A,u_1)$ and $(A,v_1)$ restricted to the underneath of the line through points $u_\nu$ and $v_\mu$. Trivially, each point of $\sW^j$ sees all vertices $u_\nu,\ldots,A,\ldots,v_\mu$.

Let $\mF_z(x,y)$ be the `$z$' half-plane defined by the line through $x$ and $y$ where `$z$' is `b' (bottom), `r' (right), or `l' (left). 
If $u_i$ and $v_l$ are distinct vertices, according to Property~\ref{prop:11}, at least one of the pairs $(u_{i+1},v_{l-1})$ and $(u_{i-1},v_{l+1})$ do not see each other. The invisible pair is determined by applying Corollary~\ref{cor:5} and Property~\ref{prop:11}. 

Assume that $(u_{i+1},v_{l-1})$ is the invisible pair. Then, $\sW^j_{i,l}$ is defined to be $\mF_r(s_{i+1},u_i)\bigcap\mF_r(v_l,u_i)\bigcap\mF_l(u_{i-1},u_i)\bigcap\mF_l(v_{l-1},u_i)\bigcap\mF_b(v_\mu,u_\nu)$
~(Fig.~\ref{fig:pseudo_W_il}). As defined in Section~\ref{sec:Tower}, $s_{i+1}$, $u_i$ and $u_{i+1}$ are collinear. We used $\mF_r(s_{i+1},u_i)$ instead of $\mF_r(u_{i+1},u_i)$ here because at least for $i=\nu$ we do not know the position of $u_{i+1}$ yet. 

Any one of these half-planes forces some visibility constraints for $w_j$. $\mF_b(v_\mu,u_\nu)$ implies that $w_j$ sees both $u_nu$ and $v_\mu$; $\mF_r(s_{i+1},u_i)$ implies that $w_j$ sees all vertices $<u_i,\ldots,u_\nu>$; $\mF_r(v_l,u_i)$ implies that $w_j$ sees all vertices $<v_l,\ldots,v_\mu>$; $\mF_l(u_{i-1},u_i)$ prevents $w_j$ from seeing vertices $<A,\ldots,u_{i-1}>$; and $\mF_l(v_{l-1},u_i)$ prevents $w_j$ from seeing vertices $<A,\ldots,v_{l-1}>$. Therefore, all points in this region satisfy the visibility constraints from $w_j$ to vertices $<u_\nu,\ldots,A,\ldots,v_\mu>$.

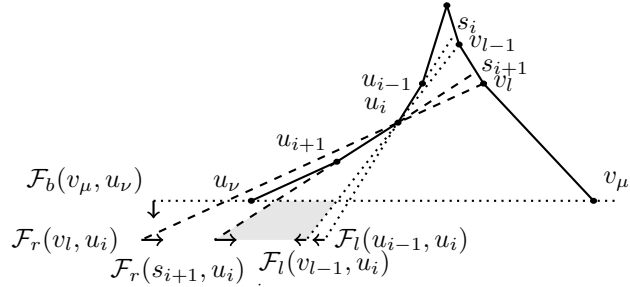
\begin{figure}
\begin{center}
\begin {tikzpicture}[thick,scale=0.65,y=0.8cm]
\draw (-1.25,0) -- (0,1) -- (0.5,2)  (1.25,3) -- (1.75,2) --  (4,-1) (0.5,2) -- (1,4) -- (1.25,3) (-1.25,0) -- (-3,-1);
\draw [dashed] (0,1) -- (1.6,2.3) (-1.25,0) -- (-3.75,-2) (1.75,2) -- (-5.25,-2);
\draw [dotted] (0.5,2) -- (1.15,3.25) (0,1) -- (-1.5,-2) (1.25,3) -- (-1.875,-2) (-5,-1) -- (4.5,-1);
\draw [fill] (-1.25,0) circle [radius=.05];
\node [above left] at (-1.25,0) {$u_{i+1}$};
\draw [fill] (0,1) circle [radius=.05];
\node [above left] at (0,1) {$u_i$};
\draw [fill] (0.5,2) circle [radius=.05];
\node [left] at (0.5,2) {$u_{i-1}$};
\draw [fill] (1,4) circle [radius=.05];
\draw [fill] (1.25,3) circle [radius=.05];
\node [right] at (1.25,3) {$v_{l-1}$};
\draw [fill] (1.75,2) circle [radius=.05];
\node [right] at (1.75,2) {$v_l$};
\draw [fill] (-3,-1) circle [radius=.05];
\node [above left] at (-3,-1) {$u_\nu$};
\draw [fill] (4,-1) circle [radius=.05];
\node [above right] at (4,-1) {$v_\mu$};
\node [right] at (1.5,2.4) {$s_{i+1}$};
\node [right] at (1,3.5) {$s_i$};
\draw [->] (-5,-1) -- (-5,-1.45);
\draw [->] (-5.25,-2) -- (-4.8,-2);
\draw [->] (-3.75,-2) -- (-3.3,-2);
\draw [->] (-1.875,-2) -- (-2.125,-2);
\draw [->] (-1.5,-2) -- (-1.75,-2);
\node [left] at (-5.5,-2) {$\mF_r(v_l,u_i)$};
\node [below] at (-4.5,-2.2) {$\mF_r(s_{i+1},u_i)$};
\node [right] at (-1.5,-2) {$\mF_l(u_{i-1},u_i)$};
\node [below] at (-1.5,-2) {$\mF_l(v_{l-1},u_i)$};
\node [left] at (-5,-0.5) {$\mF_b(v_\mu,u_\nu)$};
\path [fill=black,opacity=0.1] (-3.75,-2) -- (-2.5,-1) -- (-1.25,-1) -- (-1.875,-2);
\end{tikzpicture}
\end{center}
\vspace{-7mm}
\caption{$\sW^j_{i,l}$ is the shaded region.}
\label{fig:pseudo_W_il}
\end{figure}

Concavity of $\mU$ and $\mV$ implies that intersections $\mF_r(s_{i+1},u_i)\bigcap\mF_l(u_{i-1},u_i)$ and $\mF_r(v_l,u_i)\bigcap\mF_l(v_{l-1},u_i)$ are not empty. Therefore, $\sW^j_{i,l}$ will be empty only when $\mF_r(s_{i+1},u_i)\bigcap\mF_l(v_{l-1},u_i)$ is empty or $\mF_r(v_l,u_i)\bigcap\mF_l(u_{i-1},u_i)$ is empty. The first case is impossible, because otherwise, $u_{i+1}$ must be visible from $v_{l-1}$ which is in contradiction with invisibility assumption of $(u_{i+1},v_{l-1})$. The second case is also impossible, because then, the pair $u_i$ and $v_l$ must be invisible. But, according to Property~\ref{prop:7}, $u_i$ and $v_l$ must be visible from each other. 

Therefore, the region $\mF_r(s_{i+1},u_i)\bigcap\mF_r(v_l,u_i)\bigcap\mF_l(u_{i-1},u_i)\bigcap\mF_l(v_{l-1},u_i)$ is nonempty and some part of this intersection lies in half-plane $\mF_b(v_\mu,u_\nu)$.

According to the above discussion, $\sW^j$ is defined by $\mF_b(v_\mu,u_\nu)$ and two half-planes of $\{\mF_r(s_{i+1},u_i),\mF_r(v_l,u_i),\mF_l(u_{i-1},u_i),\mF_l(v_{l-1},u_i)\}$. The apex of $\sW^j$ is defined to be the intersection of the corresponding lines of these two half-planes which is $u_i$.

The above discussions was for the assumption that $(u_{i+1},v_{l-1})$ is the invisible pair. The description for the cases where $(u_{i-1},v_{l+1})$ is the invisible pair is symmetric: $\sW^j_{i,l}$ is $\mF_l(r_{l+1},v_l)\bigcap\mF_l(v_l,u_i)\bigcap\mF_r(v_{l-1},v_l)\bigcap\mF_r(u_{i-1},v_l)\bigcap\mF_b(v_\mu,u_\nu)$ and the apex of $\sW^j$ will be $v_l$.

If the apex of $\sW^j$ lies on $\mU$, Property~\ref{prop:11} implies that the apex of $\sW^{j-1}$ will lie on $\mU$ as well. Furthermore, Corollary~\ref{cor:5} implies that $\sW^{j-1}$ is either completely coinciding $\sW^j$ or is completely on its left.
Similarly, if the apex of $\sW^j$ lies on $\mV$, then the apex of $\sW^{j+1}$ lies on $\mV$ as well, and $\sW^{j+1}$ is either coinciding $\sW^j$ or is completely on its right.

Then, we can place the vertices $w_M,\ldots,w_N$ of $\mY$ on an arbitrary concave chain inside $\mF_b(v_\mu,u_\nu)$ in such a way that $w_j\in\sW^j$. This placement satisfies the visibility constraints for $\mX$ and $\mY$. However, to guarantee the reconstruction of $\mZ$ and $\mZ'$, we define some constraints on this concave chain which is described in the rest of this section.

Let $s'_i$ ($i>\nu$) be the intersection of $\mV$ and the line through $u_i$ and $LV^B(u_i,\mW)$, 
$r'_k$ ($k>\mu$) be the intersection of $\mU$ and the line through $v_k$ and $FV^B(v_k,\mW)$,
$t'_j$ ($j<M$) be the intersection of $\mV$ and the line through $w_j$ and $w_{j+1}$, and
$t'_j$ ($j>N$) be the intersection of $\mU$ and the line through $w_j$ and $w_{j-1}$ (see Fig.~\ref{fig:pseudo_t_r_s}). 

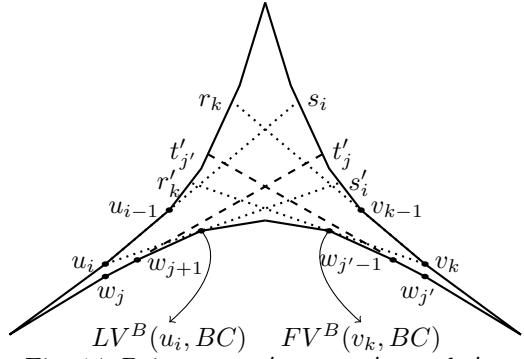
\begin{figure}
\begin{center}
\begin{tikzpicture}[thick,scale=0.85,y=0.65cm]
\draw (0,4) -- (-0.4,2) -- (-1,0) -- (-1.5,-1) -- (-2.5,-2.3) -- (-4,-4) -- (-2.5,-2.6) -- (-2,-2.2) -- (-1,-1.5) -- (0,-1.25);
\draw (0,4) -- (0.4,2) -- (1,0) -- (1.5,-1) -- (2.5,-2.3) -- (4,-4) -- (2.5,-2.6) -- (2,-2.2) -- (1,-1.5) -- (0,-1.25);
\draw [fill] (-1.5,-1) circle [radius=.04];
\node [left] at (-1.5,-1) {$u_{i-1}$};
\draw [fill] (-2.5,-2.3) circle [radius=.04];
\node [left] at (-2.5,-2.3) {$u_i$};
\draw [fill] (-2.5,-2.6) circle [radius=.04];
\node [below] at (-2.4,-2.6) {$w_j$};
\draw [fill] (-2,-2.2) circle [radius=.04];
\node [right] at (-2,-2.4) {$w_{j+1}$};
\draw [fill] (-1,-1.5) circle [radius=.04];
\draw [thin] (-1,-1.5) edge[out=-50,in=40,->] (-1.5,-3.7);
\node [below] at (-1.5,-3.5) {$LV^B(u_i,\mW)$};
\draw [fill] (1.5,-1) circle [radius=.04];
\node [right] at (1.5,-1) {$v_{k-1}$};
\draw [fill] (2.5,-2.3) circle [radius=.04];
\node [right] at (2.5,-2.3) {$v_k$};
\draw [fill] (2.5,-2.6) circle [radius=.04];
\node [below] at (2.4,-2.6) {$w_{j'}$};
\draw [fill] (2,-2.2) circle [radius=.04];
\node [left] at (2,-2.3) {$w_{j'-1}$};
\draw [fill] (1,-1.5) circle [radius=.04];
\draw [thin] (1,-1.5) edge[out=-130,in=140,->] (1.5,-3.7);
\node [below] at (1.5,-3.5) {$FV^B(v_k,\mW)$};
\node [right] at (1.15,-0.35) {$s'_i$};
\node [right] at  (0.5,1.6) {$s_i$};
\node [left] at  (-1.15,-0.35) {$r'_k$};
\node [left] at (-0.5,1.6) {$r_k$};
\node [right] at (0.9,0.35) {$t'_j$};
\node [left] at (-0.9,0.35) {$t'_{j'}$};
\draw [dotted] (-2.5,-2.3) -- (1.15,-0.35);
\draw [dotted] (-2.5,-2.3) -- (0.5,1.6);
\draw [dotted] (2.5,-2.3) -- (-1.15,-0.35);
\draw [dotted] (2.5,-2.3) -- (-0.5,1.6);
\draw [dashed] (-1.8,-2) -- (0.9,0.35);
\draw [dashed] (1.8,-2) -- (-0.9,0.35);
\end{tikzpicture}
\end{center}
\vspace{-7mm}
\caption{Points $s_{(\cdot)}$, $s'_{(\cdot)}$, $r_{(\cdot)}$, $r'_{(\cdot)}$, and $t'_{(\cdot)}$.}
\label{fig:pseudo_t_r_s}
\end{figure}

Note that although we have not yet determined positions of vertices defining $s'_i$, $r'_k$, and $t'_j$, we determine their containing edges from the visibility information as follows:
for $i>\nu$, if $u_i$ sees at least one vertex from $\mV$, $s_i$ lies on the segment connecting $P^A(FV^A(u_i,\mV))$ and $FV^A(u_i,\mV)$ and $s'_i$ lies on the segment connecting $(LV^A(u_i,\mV)$ and $N^A(LV^A(u_i,\mV))$.
On the other hand, if $u_i$ sees no vertex from $\mV$, then for $k\geq i$, both $s_k$ and $s'_k$ lie on the segment connecting $P^A(LV^A(u_j,\mV))$ and $LV^A(u_j,\mV)$ where $u_j$ has the highest index among the vertices of $\mU$ that see at least one vertex from $\mV$. Corollary~\ref{cor:3} implies that all these points lie on boundary edges of $\mX$, except when $i=\nu+1$ and $w_{M-1}$ is visible to both $u_\nu$ and $v_\mu$, for which both $s_k$ and $s'_k$ for $k\geq i$ lie on $(v_\mu,v_{\mu+1})$. The same situation happens for $r_l$ and $r'_l$ when $l>\mu$. 

The containing edge of $t'_j$ for $j<M$ is determined as follows: If $w_j$ sees at least one vertex from $\mV$, then $t'_j$ lies on the segment connecting $LV^A(w_j,\mV)$ and $N^A(LV^A(w_j,\mV))$, otherwise, it lies on the containing edge of $s'_\alpha$ 
(Note that according to our assumption at the beginning of Section~\ref{sec:reconstruction}, $\alpha$ and $\delta$ are respectively the greatest indices of vertices $u_i$ and $v_j$ on $\mU$ and $\mV$ side-chains.).
Similarly, for $j>N$, if $w_j$ sees at least one vertex from $\mU$, then $t'_j$ lies on the segment connecting $LV^A(w_j,\mU)$ and $N^A(LV^A(w_j,\mU))$, and otherwise, it lies on the containing edge of $r'_\delta$. Property~\ref{prop:11} implies that all these points lie on boundary edges of $\mX$ or edges $(v_\mu,v_{\mu+1})$ and $(u_\nu,u_{\nu+1})$.

The containing edges of $s'_\alpha$ and $r'_\delta$ are respectively called {\it``the floating edge in $\mV$''} and {\it ``the floating edge in $\mU$''}. We call these edges floating because we increase their length, and reposition their underneath vertices to enforce the concavity in building $\mZ$ and $\mZ'$.

We define the vertices $w_{M^*}$ and $w_{N^*}$ as follows: If $\mW$ has two edges, then $w_M$ and $w_N$ are both equal to $w_1$ (the middle vertex of $\mW$), and $w_{M^*}$ and $w_{N^*}$ are also defined to be $w_1$. When $\mW$ has more than two edges, $M^*$ is defined to be $M$ when the apex of $\sW^M$ does not lie on a vertex of $\mV$ below its floating edge. Otherwise, $M^*$ is defined to be $j$ where $j$ is the maximum index for which the apex of $\sW^j$ lies above the floating edge of $\mV$ (this apex may lie on $\mU$). 
If the index of $Ind^A(FV^A(w_{M^*},\mU))$ is greater than $\nu$, the apex of $\sW^{M^*}$ is temporarily assumed to be $u_\nu$ and $\sW^{M^*}$ is defined to lie between $\mF_r(s_{i+1},u_\nu)$ and $\mF_l(u_{\nu-1},u_\nu)$.
The index $N^*$ is defined similarly. It is clear that at least one of the equalities $w_{M^*}=w_M$ or $w_{N^*}=w_N$ holds.

We use $\mR(x,y)$ to denote the ray from $x$ towards $y$. In addition, $\mR_a(x,y)$ denotes the ray from $a$ and parallel to $\mR(x,y)$ (Fig.~\ref{fig:pseudo_rays}).

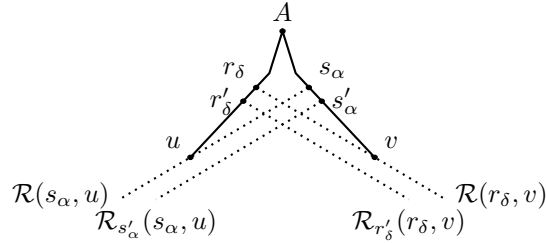
\begin{figure}
\begin{center}
\begin{tikzpicture}[thick,scale=0.7,y=0.8cm]
\draw [fill] (0,3) circle [radius=.04];
\node [above] at (0,3) {$A$};
\draw [fill] (-1.75,0) circle [radius=.04];
\node [above left] at (-1.75,0) {$u$};
\draw [fill] (1.75,0) circle [radius=.04];
\node [above right] at (1.75,0) {$v$};
\draw [fill] (0.5,1.66) circle [radius=.04];
\node [above right] at (0.5,1.66) {$s_\alpha$};
\draw [fill] (0.745,1.335) circle [radius=.04];
\node [right] at (0.745,1.335) {$s'_\alpha$};
\draw [fill] (-0.5,1.66) circle [radius=.04];
\node [above left] at (-0.5,1.66) {$r_\delta$};
\draw [fill] (-0.745,1.335) circle [radius=.04];
\node [left] at (-0.745,1.335) {$r'_\delta$};
\draw (-1.75,0) -- (-0.25,2) -- (0,3) -- (0.25,2) -- (1.75,0);
\node [left] at (-3.1,-0.9) {$\mR(s_\alpha,u)$};
\node [below] at (-2.4,-1) {$\mR_{s'_\alpha}(s_\alpha,u)$};
\node [right] at (3.1,-0.9) {$\mR(r_\delta,v)$};
\node [below] at (2.4,-1) {$\mR_{r'_\delta}(r_\delta,v)$};
\draw[dotted] (0.5,1.66) -- (-1.75,0) -- (-3.1,-1); 
\draw[dotted] (0.745,1.335) -- (-2.4,-1); 
\draw[dotted] (-0.5,1.66) -- (1.75,0) -- (3.1,-1); 
\draw[dotted] (-0.745,1.335) -- (2.4,-1); 
\end{tikzpicture}
\end{center}
\vspace{-7mm}
\caption{The rays $\mR(s_\alpha,u)$, $\mR_{s'_\alpha}(s_\alpha,u)$, $\mR(r_\delta,v)$, and $\mR_{r'_\delta}(r_\delta,v)$.}
\label{fig:pseudo_rays}
\end{figure}  

Despite our definition of the regions $\sW^i$ for all vertices $w_i\in\mW$, we refine this definition for $\sW^{N^*}$ ({\it resp.} $\sW^{M^*}$) when $N^*\neq N$ ({\it resp.} $M^*\neq M$) or the floating edge of $\mV$ ({\it resp.} $\mU$) lies under the line through $u_\nu$ and $v_\mu$.
At most one of the floating edges lies under $\mR(u_\nu,v_\mu)$. Because otherwise, either $v_mu$ will see $u_{\nu+1}$ or $u_\nu$ will see $v_{\mu+1}$ which is in contradiction with the selection of $u_\nu$ and $v_\mu$. Let $v$ be a point on $\mR_{v_\mu}(r_{\mu+1},v_\mu)$ when the floating edge of $\mV$ lies under $\mR(u_\nu,v_\mu)$, or be $v_\mu$ otherwise. Similarly, $u$ is defined to be either $u_\nu$ or a point on $\mR_{u_\nu}(s_{\nu+1},u_\nu)$. The regions $\sW^{N^*}$ and $\sW^{M^*}$ are restricted to lie under the line through $u$ and $v$. Moreover, we know that at most one of the indices $M^*$ and $N^*$ is not equal to its corresponding index $M$ or $N$. Without loss of generality, assume that $N^*\neq N$. Then, we additionally restrict the region $\sW^{N^*}$ as follows (this restriction is not applied when we reconstruct $\mZ$ or $\mZ'$).
Let $p$ be a point inside the intersection of $\sW^N$ and $\mF_b(u,v)$ and with an arbitrary positive distance from $\mR(u,v)$. We determine $t'_{N^*}$ on its edge and with $\epsilon l$ distance above the lower endpoint of this edge where $\epsilon>0$ and $l$ is the number of vertices in $\mV$ and $\mW$ whose $r'_{(\cdot)}$'s and $t'_{(\cdot)}$'s lie on this edge. The region $\sW^{N^*}$ is restricted to lie under the line through $t'_{N^*}$ and $p$ (see Fig.~\ref{fig:w_regions}).

Let $s_\alpha$ be a point on its edge and with $\epsilon k$ distance below the upper endpoint of this edge where $\epsilon>0$ and $k$ is the number of vertices in $\mU$ whose $s_{(\cdot)}$'s lie on this edge. Similarly, let $s'_\alpha$ be a point on its edge and with $\epsilon m$ distance above the lower endpoint of this edge where $\epsilon>0$ and $m$ is the number of vertices in $\mU$ and $\mW$ whose $s'_{(\cdot)}$'s and $t'_{(\cdot)}$'s lie on this edge. The value of $\epsilon$ is small enough such that $s_\alpha$ lies above $s'_\alpha$. The points $r_\delta$ and $r'_\delta$ are defined similarly.

As shown in Fig.~\ref{fig:pseudo_rays}, let $\mS$ ({\it resp.} $\mT$) be the strip defined by the supporting lines of $\mR(s_\alpha,u)$ and $\mR_{s'_\alpha}(s_\alpha,u)$ ({\it resp.} $\mR(r_\delta,v)$ and $\mR_{r'_\delta}(r_\delta,v)$).

\begin{lemma}\label{lem:2}
It is always possible to enlarge the floating edges of $\mV$  and $\mU$ and re-position the vertices which lie under the enlarged edges such that $\sW^{M^*}\bigcap\mS$ and $\sW^{N^*}\bigcap\mT$ are not empty and the new position of vertices of $\mX$ satisfy their visibility relations in the visibility graph.
\end{lemma}
\begin{proof}
Assume that the intersection of $\sW^{M^*}$ and $\mS$ is empty. According to the definition of $M^*$, the apex of $\sW^{M^*}$ either lies above the floating edge of $\mV$ or lies on $\mU$.  
This implies that enlarging the floating edge of $\mV$ only affects half-plane $\mF_b(v_\mu,u_\nu)$ that defines up-side of $\sW^{M^*}$. 
Then, we can enlarge the floating edge of $\mV$ in such a way that the lower defining ray of $\mS$ and the upper defining half-plane of $\sW^{M^*}$ intersect inside $\sW^{M^*}$ which means that the intersection of $\sW^{M^*}$ and $\mS$ is not empty. Moreover, when this intersection is not empty, this extension will just increase the intersection. On the other hand, enlarging this edge changes the position of vertices of $\mX$ which lie under this edge. For these vertices, we have their corresponding points $r$'s. By enlarging the floating edge of $\mV$, the new positions will be computed according to their definition (for a vertex $v_i$ it must lie on the supporting line of $r_i$ and $v_{i-1}$) to satisfy the visibility relations in the visibility graph reduced to vertices of $\mX$. To complete the proof, it is simple to see that extending the floating edge of $\mU$ will again increase the intersection of $\mW^{M^*}$ and $\mS$.

The proof for $w_{N^*}$ is analogously the same. \qed
\end{proof}

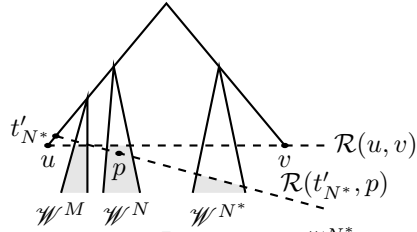
\begin{figure}
\begin{center}
\begin{tikzpicture}[thick,y=0.6cm,scale=0.7]
\draw (-2.25,-0.5) -- (0,4) -- (2.25,-0.5);
\draw (1.5,-2) -- (1,2) -- (0.5,-2);
\draw (-1.2,-2) -- (-1,2) -- (-0.5,-2);
\draw (-1.5,-2) -- (-1.5,1) -- (-2,-2);
\draw [dashed] (3,-0.5) -- (-2.25,-0.5);
\node [right] at (3,-0.5) {$\mR(u,v)$};
\draw [dashed] (-2.1,-0.2) -- (3,-2.5);
\node [above] at (3.2,-2.5) {$\mR(t'_{N^*},p)$};
\draw [fill] (-2.25,-0.5) circle [radius=.05];
\node [below] at (-2.25,-0.5) {$u$};
\draw [fill] (2.25,-0.5) circle [radius=.05];
\node [below] at (2.25,-0.5) {$v$};
\draw [fill] (-2.1,-0.2) circle [radius=.05];
\node [left] at (-2,0) {$t'_{N^*}$};
\draw [fill] (-0.9,-0.75) circle [radius=.05];
\node [below] at (-0.9,-0.75) {$p$};
\path [fill=black,opacity=0.1] (0.5,-2) -- (0.55,-1.4) -- (1.5,-1.8) -- (1.5,-2);
\node [below] at (-2,-2) {$\sW^{M}$};
\path [fill=black,opacity=0.1] (-1.2,-2) -- (-1.12,-0.5) -- (-0.7,-0.5) -- (-0.5,-2);
\node [below] at (-0.75,-2) {$\sW^{N}$};
\path [fill=black,opacity=0.1] (-2,-2) -- (-1.75,-0.5) -- (-1.5,-0.5) -- (-1.5,-2);
\node [below] at (1,-2) {$\sW^{N^*}$};
\end{tikzpicture}
\end{center}
\vspace{-7mm}
\caption{Restricting $\sW^{N^*}$.}
\label{fig:w_regions}
\end{figure}  

After locating the position of vertices in $\mX$ (by possibly extending the floating edges), we place the vertices of $\mY$ as follows: If $N^*\neq N$, then we set $p$ as $w_N$ and place $w_{M^*}=w_M$ inside the intersection of $\sW^M$ and $\mS$ in such a way that both $w_M$ and $w_N$ be visible to $u$ and $v$; neither $w_M$ blocks the visibility of $w_N$, nor $w_N$ blocks the visibility of $w_M$. 
When $M^*\neq M$, $w_M$ and $w_N$ are positioned analogously. Finally, if $M^*=M$ and $N^*=N$, we select a point from $\mS\cap\sW^M$ as $w_M$ and a point from $\mT\cap\sW^N$ as $w_N$ again in such a way that both see $u$ and $v$.
Then, we put the vertices $w_{M+1},\ldots,w_{N-1}$ on a slightly concave chain from $w_M$ to $w_N$ in such a way that each $w_j$ ($M\leq j\leq N$) lies inside $\sW^j$ and sees $u$ and $v$.

Based on the definition of $\sW^i$'s regions and the specified positions of vertices inside these regions, this setting is compatible with the visibility graph restricted to the vertices of $\mX$ and $\mY$.
\vspace{-2mm}
\subsection{Reconstructing $\mZ$ and $\mZ'$}\label{sec:recontructZ}
\vspace{-2mm}
In this step, we place the vertices of $\mZ$ and $\mZ'$ to complete the reconstruction procedure. 
As said before, $\mZ$ (resp. $\mZ'$) is a part of the target pseudo-triangle with $<u_\nu,u_{\nu+1},\ldots,B,\ldots,w_M>$ (resp. $<v_\mu,v_{\mu+1},\ldots,C,\ldots,w_N>$) boundary vertices. Here, we only describe how to build $\mZ$. The construction of $\mZ'$ is symmetrically the same.

Location of a vertex $u_i\in\mZ$ is determined by the intersection point of the rays $\mR(s_i,u_{i-1})$ and $\mR(s'_i,LV^B(u_i,\mW))$ and location of a vertex $w_h\in\mZ$ is an arbitrary point on $\mR(t'_h,w_{h+1})$ inside the region $\sW^h$. Therefore, to construct $\mZ$ we start from $u_{\nu+1}$ and $w_{M-1}$, and in each step we determine the position of one of the vertices and go forward to the next vertex. This is done by incrementally determining direction of the rays $\mR(s_i,u_{i-1})$, $\mR(s'_i,LV^B(u_i,\mW))$, and $\mR(t'_h,w_{h+1})$ as well as $\sW^h$ regions.

Consider the edges of the pseudo-triangle on which the points $s_i$, $s'_i$, $r_j$, $r'_j$, and $t'_l$ for $i>\nu$, $j>\mu$, and $l<M$ and $l>N$ lie. 
Keep an upper point and a lower point for each edge. Initialize the upper point with the upper endpoint of that edge or the latest located $s_{(\cdot)}$ or $r_{(\cdot)}$ on this edge. Initialize the lower point with the lower endpoint of the edge.
Position of each $s_{(\cdot)}$, $r_{(\cdot)}$, $s'_{(\cdot)}$, $r'_{(\cdot)}$, and $t'_{(\cdot)}$ is determined whenever we need the rays passing through them. 
We place the points $s'_{(\cdot)}$, $r'_{(\cdot)}$, and $t'_{(\cdot)}$, with $\epsilon>0$ distance above the current lower point of their edges and place the points $s_{(\cdot)}$ and $r_{(\cdot)}$, with $\epsilon>0$ distance below the upper point of their edges. Whenever a new $s_{(\cdot)}$, $r_{(\cdot)}$, $s'_{(\cdot)}$, $r'_{(\cdot)}$, or $t'_{(\cdot)}$ point is located on an edge, the upper or lower point of that edge is updated properly.

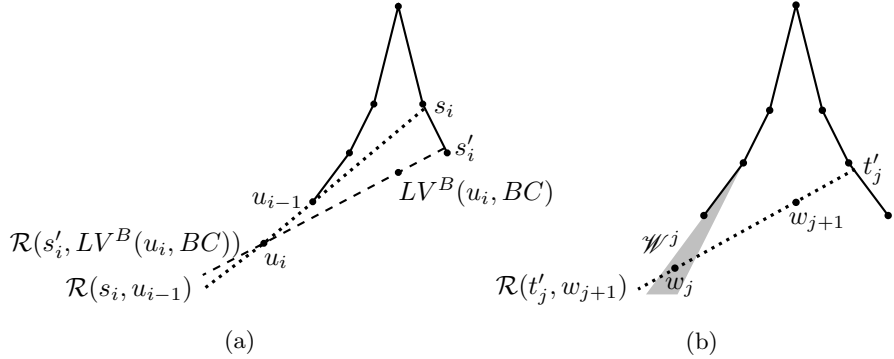
\begin{figure}
\begin{center}
\begin{subfigure}{.5\textwidth}
\begin{tikzpicture}[thick,scale=0.65]
\draw [dashed] (0.95,-0.9) -- (-2.75,-2.85) -- (-4,-3.5);
\draw [very thick,dotted] (0.525,-0.1) -- (-2.75,-2.85) -- (-4,-3.8);
\draw  (1,-1) -- (0.5,0) -- (0,2) -- (-0.5,0) -- (-1,-1) -- (-1.75,-2);
\draw [fill] (0,2) circle [radius=.05];
\draw [fill] (-0.5,0) circle [radius=.05];
\draw [fill] (-1,-1) circle [radius=.05];
\draw [fill] (-1.75,-2) circle [radius=.05];
\node [left] at (-1.75,-2) {$u_{i-1}$};
\draw [fill] (0,-1.4) circle [radius=.05];
\node [right] at (-0.2,-1.8) {$LV^B(u_i,\mW)$};
\draw [fill] (1,-1) circle [radius=.05];
\draw [fill] (0.5,0) circle [radius=.05];
\node [right] at (0.95,-0.9) {$s'_i$};
\node [right] at (0.525,-0.1) {$s_i$};
\draw [fill] (-2.75,-2.85) circle [radius=.05];
\node [below] at (-2.5,-2.85) {$u_i$};
\node [left] at (-3,-2.85) {$\mR(s'_i,LV^B(u_i,\mW))$};
\node [left] at (-4,-3.8) {$\mR(s_i,u_{i-1})$};
\end{tikzpicture} 
\caption{}
\label{fig:pseudo_u_i}
\end{subfigure}
\begin{subfigure}{.45\textwidth}
\begin{tikzpicture}[thick,scale=0.7]
\draw (0,2) -- (-0.5,0) -- (-1,-1) -- (-1.75,-2)  (1.75,-2) -- (1,-1) -- (0.5,0) -- (0,2);
\draw [fill] (0,2) circle [radius=.05];
\draw [fill] (-0.5,0) circle [radius=.05];
\draw [fill] (-1,-1) circle [radius=.05];
\draw [fill] (-1.75,-2) circle [radius=.05];
\draw [fill] (-2.3,-3) circle [radius=.05];
\node [below] at (-2.2,-3) {$w_j$};
\draw [fill] (0,-1.75) circle [radius=.05];
\node [below] at (0.4,-1.75) {$w_{j+1}$};
\draw [fill] (1.75,-2) circle [radius=.05];
\draw [fill] (1,-1) circle [radius=.05];
\draw [fill] (0.5,0) circle [radius=.05];
\node [right] at (1.09,-1.12) {$t'_j$};
\draw [very thick,dotted] (1.12,-1.12) -- (-3,-3.4);
\path [fill=black,opacity=0.25] (-1,-1) -- (-2.85,-3.5) -- (-2.25,-3.5);
\node [left] at (-3,-3.4) {$\mR(t'_j,w_{j+1})$};
\node [left] at (-2,-2.5) {$\sW^j$};
\end{tikzpicture} 
\caption{}
\label{fig:pseudo_w_j}
\end{subfigure}
\caption{(a) Determining $u_i$, (b) Determining $w_j$.}
\label{F8}
\end{center}
\end{figure}

More precisely, assume that we have already determined positions of vertices $u_\nu,u_{\nu+1},\ldots,u_{i-1}$ ($i>\nu$) as well as the vertices $w_M,w_{M-1},\ldots,w_{j+1}$ ($j<M$). To determine position of one of the vertices $u_i$ and $w_j$ we do as follows: Let $w_k$ be $LV^B(u_i,\mW)$. If $k<j$, then we have already located the position of $w_k$, and directions of the rays $\mR(s_i,u_{i-1})$ and $\mR(s'_i,LV^B(u_i,\mW))$ are known. We will show in Lemma~\ref{lem:3} that these rays intersect. So, $u_i$ is located on the intersection point of these rays (Fig.~\ref{fig:pseudo_u_i}). Otherwise, we must first determine position of $w_j$ which lies on $\mR(t'_j,w_{j+1})$ and inside $\sW^j$ (Fig.~\ref{fig:pseudo_w_j}). The position of $w_{j+1}$ is already known and $t'_{j+1}$ is determined according to the above paragraph. From these two points the direction of $\mR(t'_j,w_{j+1})$ is obtained. The region $\sW^j$ is determined as follows:
Suppose that $FV^A(w_j,\mU)=u_k$ and $FV^A(w_j,\mV)=v_l$. We define $\sW^{j}$ as in the previous section with the exception that it may be possible that only one of the vertices $u_k$ and $v_l$ exists. By Corollary~\ref{cor:4}, for $j<M$, $u_j$ always exists.
If $w_j$ sees no vertex from $\mV$, then it would see a part of the floating edge of $\mV$. Hence, we consider the upper endpoint of this edge as $v_{l-1}$. From properties~\ref{prop:7} and \ref{prop:11} we know that $\sW^j$ is not empty and lies to the left of $\sW^{j+1}$. Moreover, it will be shown in Lemma~\ref{lem:3} that $\mR(t'_j,w_{j+1})$ intersects $\mR_{s'_\alpha}(s_\alpha,u)$. Since $\mR_{s'_\alpha}(s_\alpha,u)$ passes through all $\sW^{(\cdot)}$, $\mR(t'_j,w_{j+1})$ passes through $\sW^j$.
Therefore, we can determine the position of $w_j$. 

Note that the definition of $s_i$'s and $t'_j$'s enforce the concavity of the vertices on $\mU$ and $\mW$, respectively. According to the definition of $\mR(s_i,u_{i-1})$ and $\mR(s'_i,LV^B(u_i,\mW))$ for $u_i$ and $\mR(t'_j,w_{j+1})$ and $\sW^j$ for $w_j$, in both cases (locating $u_i$ or $w_j$), visibility of the newly located vertex is exactly the same as its visibility in the visibility graph (restricted to the vertices of $\mX$, $\mY$, and the constructed part of $\mZ$). This means that at the end of this construction which vertices $B$ and $C$ are located the visibility graph of the constructed polygon is consistent with the input visibility graph.

\begin{lemma}\label{lem:3}
The rays $\mR(s_i,u_{i-1})$ and $\mR(s'_i,LV^B(u_i,\mW))$ for $i>\nu$ are convergent inside $\mS$. 
\end{lemma}
\begin{proof}
Remember that $\mS$ is the strip defined by the supporting lines of $\mR(s_\alpha,u)$ and $\mR_{s'_\alpha}(s_\alpha,u)$.
By Corollaries~\ref{cor:1} and~\ref{cor:3}, we know that $s_i$ lies above the strip $\mS$ and $s'_i$ lies below this strip. Then, it is enough to show that for $i>\nu$, $\mR(s'_i,LV^B(u_i,\mW))$ crosses $\mR_{s'_\alpha}(s_\alpha,u)$ and $\mR(s_i,u_{i-1})$ crosses $\mR(s_\alpha,u)$. We first prove that $\mR_{s'_\alpha}(s_\alpha,u)$ intersects $\mR(s'_i,LV^B(u_i,\mW))$.  
Let $LV^B(u_i,\mW)=w_h$. For $M^*\leq h\leq M$, it can be easily shown by induction that $w_h$ is located above $\mR(t'_{M^*},w_M)$. Moreover, it is simple to see that $s'_i$ must lie below $t'_{M^*}$. Then, knowing that $\mR(t'_{M^*},w_M)$ crosses $\mR_{s'_\alpha}(s_\alpha,u)$ implies that $\mR(s'_i,w_h)$ intersects $\mR_{s'_\alpha}(s_\alpha,u)$ as well. From the fact that $w_{M^*}$ lies inside $\mS$, it can also be shown by induction that $w_h$ for $h<M^*$ lies inside $\mS$ which means that $\mR(s'_i,w_h)$ crosses $\mR_{s'_\alpha}(s_\alpha,u)$.

To complete the proof, we prove by induction on $i$ that $\mR(s_i,u_{i-1})$ crosses $\mR(s_\alpha,u)$. It is clear that $s_{\nu+1}$ is located above $s_\alpha$ which means that $\mR(s_{\nu+1},u_\nu)$ intersects $\mR(s_\alpha,u)$.
From the previous paragraph we know that $\mR_{s'_\alpha}(s_\alpha,u)$ intersects $\mR(s'_{\nu+1},LV^B(u_{\nu+1},\mW))$. Therefore, $\mR(s_{\nu+1},u_\nu)$ and $\mR(s'_{\nu+1},LV^B(u_{\nu+1},\mW))$ will intersect at a point within $\mS$. Since we put $u_{\nu+1}$ at this intersection point, as the induction step, assume that $u_{i-1}$ lies inside $\mS$ where $i>\nu+1$. Then, $\mR(s_i,u_{i-1})$ intersects $\mR(s_\alpha,u)$. \qed
\end{proof}

\vspace{-2mm}
\section{Analysis}
\vspace{-2mm}
In previous sections, we proved several properties on the visibility graph of a pseudo-triangle and proposed an algorithm that constructs a pseudo-triangle for a given pair of visibility graph $\mG(V,E)$ and Hamiltonian cycle $\mH$ when this pair supports these properties. In this section, we analyze the time complexity of algorithms required to check these properties and the running time of the reconstruction algorithm.

To check Property~\ref{prop:1}, we need a linear time trace on vertices of $\mG$ according to their order in $\mH$. This is done in $\mO(n)$ time. If two corners are identified in this way, the existence of a tower polygon corresponding to the pair of $\mG$ and $\mH$ can also be verified in linear time~\cite{CoLuSp97}. 
Property~\ref{prop:2} can be verified in $\mO(|E|)$ by a simple trace of the edge list of the visibility graph. Precisely, for each vertex $u\in\mU$ we maintain the minimum index, maximum index, and number of vertices of the other side-chains $\mV$ and $\mW$ which are visible from $u$. After finishing this trace, from these triple of parameters (minimum index, maximum index, number of visible vertices) the Property~\ref{prop:2} is checked in $\mO(n)$ time. In order to verify the rest of the properties, it is required to know the visible subchains from each vertex. These subchains are obtained as by-products using the method proposed for checking Property~\ref{prop:2}. Having these subchains for each vertex, Property~\ref{prop:3} can be verified in $\mO(n)$. 

In Property~\ref{prop:689}, for each pair $(\mU,\mV)$ of side-chains, we must find all pairs of visible vertices $(u\in\mU,v\in\mV)$ such that $(N^A(u),v)$ and $(u,N^A(v))$ are invisible. Having the visible subchains for each vertex, this check is done in constant time for each edge $(u,v)\in E$. Therefore, all pairs of vertices $(u\in\mU,v\in\mV)$ satisfying assumption of this property can be obtained in $\mO(|E|)$. Then, for each pair the three necessary conditions are checked in constant time using the maintained visible subchains of $u$ and $v$ vertices.
Checking properties~\ref{prop:7},~\ref{prop:10} and~\ref{prop:11} can be done in $\mO(n)$ by simple trace on the side-chains. Therefore, all properties can be verified in $\mO(|E|)$.

To complete the analysis, we compute the running time of the reconstruction algorithm presented in Section~\ref{sec:reconstruction}. Assume that $\mG$ satisfies all of the properties introduced in Section~\ref{sec:Properties} and we know the visible subchains of each vertex according to their order in $\mH$. 
The side-chains of the target pseudo-triangle are identified in linear time according to the algorithm described in the proof of Lemma~\ref{lem:1}.
Reconstructing $\mX$ is done using the tower reconstruction algorithm whose running time is linear in terms of the number of edges in the visibility graph reduced to $\mX$. To reconstruct $\mY$, the algorithm needs to determine the floating edges of $\mU$ and $\mV$ which can be done in constant time. Computing the $\sW$-type regions (for each vertex $w_i\in\mW$) and determining the vertices $w_{N^*}$ and $w_{M^*}$ needs $\mO(n)$ time. If the conditions of Lemma~\ref{lem:2} are not satisfied, the floating edges of $\mU$ and $\mV$ must be extended which is done in $\mO(1)$: A lower bound for the increase in floating edges can be computed by using Thales' theorem and trigonometric functions. Locating each vertex of $\mY$ is also done in constant time. Finally, placing each vertex of $\mZ$ and $\mZ'$ takes constant time, as well. Therefore, the total running time of the algorithm is $\mO(|E|)$. We can combine all results as:

\begin{theorem}
The visibility graph and the boundary vertices of a pseudo-triangle satisfy properties~\ref{prop:1} to \ref{prop:11}, and conversely, for any pair of graph $\mG$ and Hamiltonian cycle $\mH$ satisfying these properties, there is a pseudo-triangle $\mP$ whose visibility graph and boundary vertices are respectively isomorphic to $\mG$ and $\mH$. Checking these properties and reconstructing such a polygon can be done in $\mO(|E|)$.
\end{theorem}

\section{Conclusion}\label{sec:Conclusion}

In this paper, we considered properties of the visibility graph of a pseudo-triangle and obtained a set of necessary and sufficient conditions that such graphs must have. Then, we propose an algorithm to reconstruct a polygon from a given visibility graph which supports these properties. This characterizing and reconstructing problem, despite its long history, is still at the start of its way to be solved for all polygons.

\bibliographystyle{elsarticle-num}
\bibliography{bibliography}

\end{document}